\documentclass[superscriptaddress,twocolumn,floatfix,nofootinbib,prx]{revtex4-2}

\usepackage{amsmath}
\usepackage{amssymb}
\usepackage{mathtools}
\usepackage{graphicx}
\usepackage{braket}
\usepackage{physics}
\usepackage[dvipsnames]{xcolor}
\usepackage[colorlinks=true,linkcolor=RoyalBlue,citecolor=ForestGreen,urlcolor=RoyalPurple]{hyperref}
\usepackage[capitalise]{cleveref}
\usepackage{dsfont}
\usepackage{placeins}
\usepackage{cancel}
\usepackage[normalem]{ulem}
\usepackage{dcolumn}
\usepackage{bm}
\usepackage{amsfonts,amsthm, bbm}
\usepackage{bbold}
\usepackage{comment}
\usepackage{xfrac}
\renewcommand{\Im}{\mathrm{Im}}
\renewcommand{\Re}{\mathrm{Re}}

\usepackage{makerobust}

\newcommand{\ie}{{\it i.e.}}

\definecolor{myblue}{rgb}{0,0,0.75}

\newtheorem{theorem}{Theorem}
\newtheorem{definition}{Definition}
\newtheorem{lemma}{Lemma}

\begin{document}

\title{\textit{g}-factor symmetry and topology in semiconductor band states}

\author{Mira Sharma}
\affiliation{%
 Institute for Quantum Information, RWTH Aachen University, D-52056 Aachen, Germany\\
}%

\author{David P. DiVincenzo}
\affiliation{
 Peter Gr{\"u}nberg Institute, Theoretical Nanoelectronics, Forschungszentrum J{\"u}lich, D-52425 J{\"u}lich, Germany}%

\date{\today}

\begin{abstract}
The $\bm{g}$ tensor, which determines the reaction of Kramers-degenerate states to an applied magnetic field, is of increasing importance in the current design of spin qubits. It is affected by details of heterostructure composition, disorder, and electric fields, but it inherits much of its structure from the effect of the spin-orbit interaction working at the crystal-lattice level. Here we uncover new symmetry and topological features of $\bm{g}=\bm{g}_L+\bm{g}_S$ for important valence and conduction bands in silicon, germanium, and gallium arsenide. For all crystals with high (cubic) symmetry, we show that large departures from the nonrelativistic value $g=2$ are {\em guaranteed} by symmetry. In particular, considering the spin part $\bm{g}_S(\bm{k})$, we prove that the scalar function $det(\bm{g}_S(\bm{k}))$ must go to zero on closed surfaces in the Brillouin zone, no matter how weak the spin-orbit coupling is. We also prove that for wave vectors $\bm{k}$ on these surfaces, the Bloch states $\ket{u_{n\bm{k}}}$ have maximal spin-orbital entanglement. Using tight-binding calculations, we observe that the surfaces $det(\bm{g}(\bm{k}))=0$ exhibit many interesting topological features, exhibiting Lifshitz critical points as understood in Fermi-surface theory. 
\end{abstract}

\maketitle


\section{\label{sec:level1}Introduction}
Semiconductor spin qubits hosted in quantum dots \cite{PhysRevA.57.120} are emerging as promising candidates for scalable quantum computing due to their long coherence times and potential integration into current semiconductor technology. In order to do coherent spin manipulation it is important to identify noise sources. There are two main sources of noise in spin qubits: charge noise and spin noise. Charge noise can result in spin dephasing due to spin-orbit interactions \cite{PRXQuantum.4.020305,Kuhlmann_2013} and its dependence on the $g$-factor. The role of the $g$-factor in the realization of scalable and robust information processing using spin qubits is paramount. Even in a device structure like a quantum dot, the {\em g}-factor, an intrinsic one-electron property that characterizes the magnetic moment of Kramers-degenerate electronic states, is largely determined by the underlying crystal physics, as explored in the present paper.

In fundamental particle physics, the electron $g$-factor, as found from Dirac's relativistic wave-equation, is a dimensionless magnetic moment which has an exact value of 2. Quantum electrodynamics modifies the Dirac $g$-factor, resulting in very small deviations from 2. In the solid state, relativistic interactions in many cases cause the $g$-factor value to differ substantially from 2. Band electrons are subject to a crystal potential $V(r)$ which by itself does not cause any changes of the $g$-factor. But the spin-orbit interaction, working together with the crystal potential, modifies the $g$-factor. Since these deviations depend on $V(r)$, they are crystal and band specific, unlike the deviations coming from quantum electrodynamics, which are much smaller and universal in nature. Our investigation in this paper provides a theoretical framework, using the tight-binding formalism, to describe the $g$-factor in the setting of bulk semiconductor materials.

The effective mass tensor is an important quantity for the determination of various properties in band theory. Luttinger provided a full description of the effects of magnetic field on a band state, remarking that the inverse effective mass is not isotropic, and not even symmetric \cite{PhysRev.102.1030}. He showed that the antisymmetric part of the inverse effective mass describes the orbital response of the band electrons to an external magnetic field, resulting in an orbital contribution $\bm{g}_L$, where we see that the $g$-factor is itself promoted to the form of a tensor (but not necessarily a symmetric one). This is distinct from the spin contribution $\bm{g}_S$ arising from the Zeeman interaction. Thus one has a formalism for computation of the full response $\bm{g}_{tot} = \bm{g}_L + \bm{g}_S$; we will compute and discuss interesting features of this quantity in this paper.

It is interesting to note that another analysis for $\bm{g}_L$, on the face of it different from Luttinger's theory, came to light much later in the context of Berry-phase theory \cite{RevModPhys.82.1959}. Here one considers a localized wavepacket composed of band states. It is noted that the circulating currents inside the wavepacket are a source of a magnetic moment, whose value is shown to be independent of the detailed shape of the wavepacket, and given by a formula involving the band Berry curvature. Below we provide a simple derivation showing that this formula is in fact equivalent to the Luttinger approach involving the antisymmetric effective mass.

One might expect that in materials where spin-orbit coupling is weak, with crystalline silicon as a possible example, we might expect a $g$-factor always close to 2, which would mean more specifically $\bm{g}_S$ having three eigenvalues close to 2, and $\bm{g}_L$ having very small eigenvalues. While this is true for certain bands in certain parts of the Brillouin zone, we find the actual situation more surprising. 

We find that symmetry and topology guarantees that values of $\bm{g}_S$ far from 2, and in fact reaching to zero, {\em must} occur for certain states. We show that for bands belonging to certain irreducible representations, there exists a surface in k-space for which $det(\bm{g}_S)=0$, and that these states must have rigorously maximal spin-orbital entanglement.

Thus, we find an interesting new phenomenology of the $g$-factor, which we explore for the semiconductors silicon (Si), germanium (Ge), and gallium arsenide (GaAs). Many bands have $det(\bm{g}_S)=0$ surfaces, as well as surfaces of $det(\bm{g}_{tot})=0$. Our tight-binding calculations show that these surfaces come in many different forms and topologies, reminiscent of Fermi surfaces. We note that Lifshitz critical points, an exotic feature in Fermiology, occur rather generically in our surfaces.

\section{Theory of the \texorpdfstring{$\bm{g}$}{TEXT}-tensor}\label{sec:gtheory}

There have been numerous studies on the determination of the $g$-factor both theoretically and experimentally. The $g$-factor serves as a quantity of interest for many applications such as device spectroscopy or g-Tensor Modulation Resonance. It is crucial in order to understand spin effects. An expression for the $g$-tensor has been derived using the effective mass approximation in the \textit{k $\cdot$ p} model \cite{PhysRev.118.1534,PhysRevLett.6.683,app9.1,Fraj2007EffectiveLF}. The total Hamiltonian of an electron in an external magnetic field is
\begin{align}
    H(\bm{B}) &= \frac{1}{2m}\left[ \bm{p}-\frac{|e|\bm{A}}{c}\right]^2 + V(\bm{r}) \notag \\
    &+ \frac{\hbar}{4c^2m^2}(\nabla V(\bm{r})\times \bm{p})\cdot \bm{\sigma}    + \frac{1}{2}g\mu_B\bm{\sigma}\cdot\bm{B},
\end{align}
where $\bm{A}$ is the vector potential such that $\bm{B} = \nabla \times \bm{A}$. We identify the spin-orbit Hamiltonian,
\begin{equation}
    H_{SO} = \frac{\hbar}{4c^2m^2}\left(\nabla \bm{V} \times \bm{p}\right)\cdot \bm{\sigma}=\frac{\alpha^2}{4}\left(\nabla \bm{V} \times \bm{p}\right)\cdot \bm{\sigma}.\label{eq:AH} 
\end{equation}
The last part is written in atomic Hartree units, introducing the fine structure constant $\alpha$. Note that the momentum operator in the presence of spin-orbit interaction is
\begin{equation}
    \bm{\pi} = \bm{p}+\frac{\hbar}{4mc^2}\bm{\sigma} \times \nabla V(\bm{r}).
\end{equation}
The external field changes the crystal momentum $\bm{k}$ to $\bm{k} - \frac{|e|}{\hbar c}\bm{A}$. The components of the crystal momentum do not commute. As first noted  by Luttinger \cite{PhysRev.102.1030},
the inverse effective mass is consequently not symmetric. The antisymmetric part
\begin{equation}
\left(\frac{1}{m^*}\right)_{ij}^{AS} = \frac{1}{2}\left[\left(\frac{1}{m^*}\right)_{ij}-\left(\frac{1}{m^*}\right)_{ji}\right]
\label{eq:AS}    
\end{equation} 
is proportional to $B$. The asymmetric term has, for band $n$, the operator form \cite{YAFET19631,cardona2005fundamentals}
\begin{align}
    \left(\frac{1}{m^*}\right)_{ij}^{AS} &= \frac{1}{m^2}\sum_{n'\neq n}   
    \frac{\pi_i\ket{u_{n'k}}\bra{u_{n'k}}\pi_j}
    {E_{n}-E_{n'}} \notag \\
    &- 
    \frac{1}{m^2}\sum_{n'\neq n}\frac{\pi_j\ket{u_{n'k}}\bra{u_{n'k}}\pi_i}
    {E_{n}-E_{n'}}.\label{lutt}
\end{align}
It is understood that this antihermitian operator acts within the space of the two (spin degenerate) crystal Bloch wave functions $\ket{u_{nk}}$ of band $n$. The Zeeman splitting of band state at energy $E_{nk}$ in magnetic field $\bf{B}$ is given by the eigenvalues of the $2\times 2$ spin Hamiltonian
\begin{equation}
    H_{nk} = E_{nk} + \frac{\hbar^2}{2} \epsilon_{ijl}\left(\frac{1}{m^*}\right)^{AS}_{ij}\left( \frac{i|e|B_l}{\hbar c}\right) + \mu_B\bm{\sigma}\cdot \bf{B}
    \label{eq:Hnk}
\end{equation}
(using the Einstein summation convention).
The splitting of the spin-degenerate eigenvalues around the zero-field value $E_{nk}$ has the form
\begin{equation}
    \Delta E_{nk}=\mu_B\sqrt{B_iG_{ij}B_j}.\label{eq:DE}
\end{equation}
This defines the symmetric tensor $G_{ij}$, from which the band $g$-tensor is obtained \footnote{Note that we use the convention of Abragam and Bleaney \cite{Abragam1970ElectronPR} in Eq.\ref{eq:Gg}; other, more recent work (e.g.,  \cite{PhysRevX.8.021068}) uses a transposed convention, i.e., $G=g^Tg$.}:
\begin{equation}
    \bm{G} = \bm{g}\cdot \bm{g}^T.
    \label{eq:Gg}
\end{equation}
Note that Eq.\ref{eq:Gg} does not uniquely define $\bm{g}$, leaving open the question of whether it should be considered symmetric or not. As we will see below, $\bm{g}$ can also be defined in terms of the magnetization of the band state, but then there is still non-uniqueness arising from the choice of basis of the pair of band states $\ket{u_{nk}}$. Some previous work has considered it natural to take the $g$-tensor to be non-symmetric, see  \cite{PhysRevX.8.021068,PhysRevLett.113.236401}. We will adopt a convention for which $\bm{g}$ is not necessarily symmetric; but our discussion will depend only on the singular values of $\bm{g}$, which are uniquely determined by definition Eq.\ref{eq:Gg}.\\
\\
Equation \ref{eq:Hnk} shows that $\bm{g}$ will have two distinct contributions, $\bm{g}=\bm{g}_L+\bm{g}_S$. If spin orbit is zero, only the second, pure spin term will contribute, in this case with the value 2. For finite spin orbit its contribution, which we denote $\bm{g}_S$, will in general be modified, and there will also be a nonzero orbital magnetization contribution $\bm{g}_L$ from the first term of Eq.\ref{eq:Hnk}. $\bm{g}_L$ is directly determined by the antisymmetric part of the inverse effective mass tensor. \\

\subsection{The spin \texorpdfstring{$\bm{g}$}{TEXT}-tensor, \texorpdfstring{$\bm{g}_S$}{TEXT}}

The Bloch states $\ket{u_{nk}}$ of band $n$ of the last section consist of two orthogonal partners, which we will denote as  eigenvectors $\ket{\xi}$ and $\ket{\bar{\xi}}$. In most cases the relations between these two states involves the time-reversal operation, \ie, a complex conjugation, $\kappa$, followed by a $\sigma_y$ Pauli rotation,
\begin{equation}
    \ket{\bar{\xi}} = i\sigma_y\kappa \ket{\xi}.
\end{equation}
The states $\ket{u_{nk}}$ are generally a pseudo-Kramers pair, as the time reversal actually relates $\ket{u_{nk}}$ and $\ket{u_{n,-k}}$. But $\ket{\xi}$ and $\ket{\bar{\xi}}$ will here always denote the pair of states at wavevector $\bm{k}$ that are spin-degenerate in the limit of zero spin-orbit interaction. ``Pseudo-Kramers" means, in many cases of interest, that these states are related by time reversal followed by space inversion.\\
\\
Introducing the spin matrix
\begin{equation}\label{eq:s_mat}
S_i = \frac{1}{2}
\begin{pmatrix}
  \expval{\xi|\sigma_i|\xi} & \expval{\xi|\sigma_i|\bar{\xi}} \\
  \expval{\bar{\xi}|\sigma_i|\xi} & \expval{\bar{\xi}|\sigma_i|\bar{\xi}}
\end{pmatrix},
\end{equation}
one gets a simple expression for the spin contribution to the $g$-tensor \cite{{Abragam1970ElectronPR}} 
\begin{equation}\label{eq:gs}
    g_{Sij} = \frac{1}{2} \Tr(2S_i\tilde\sigma_j).
\end{equation}
The appearance of 2$S_i$ in Eq.\ref{eq:gs} results in the correct spin contribution to the $g$-factor due to the magnetic moment formula $\bm{M} = \bm{L} + 2\bm{S}$. Here and below in Eq.\ref{eq:gm}, $\tilde\sigma_j$ denotes a Pauli-matrix operator acting on the space of the states $\ket{\xi}$, $\ket{\bar{\xi}}$, not to be confused with the operators $\sigma_j$ above, which operate on the electron-spin Hilbert space $\ket{\uparrow}$, $\ket{\downarrow}$.

\subsection{The orbital moment contribution to the \texorpdfstring{$\bm{g}$}{TEXT}-tensor, \texorpdfstring{$\bm{g}_L$}{TEXT}}

It is interesting to note that since the original work of Luttinger, a different point of view has emerged about the magnetic moment of band electrons. One considers the electron state as the coherent transport of a wave-packet moving in $\bm{k}$-space \cite{PhysRevB.72.085110,PhysRevB.53.7010,RevModPhys.82.1959}. Since the spread of the wave-packet is taken to be much smaller than the dimension of the Brillouin zone, it is affected only by the local band properties. Due to a rotation around its center of mass, this wavepacket has an orbital magnetic moment given by\footnote{The sign of Eq.\ref{eq:M} is the opposite of the sign found in the general literature. The sign of Eq.\ref{eq:mag_mm} corresponds to the atomic limit of the Land{\'e} g.}
\begin{equation}\label{eq:M}
    L(\bm{k}) = i\frac{e}{2\hbar}\expval{\nabla_{\bm{k}} u_n| \times \left[H(\bm{k})-\epsilon_n(\bm{k})\right]|\nabla_{\bm{k}}u_m},
\end{equation}
where we see the appearance of the Berry curvature $\ket{\nabla_{\bm{k}}u_m}$.\\
\\
This alternative formulation is completely equivalent to the Luttinger result Eq.\ref{eq:AS}, which we show by the following steps.
The gradient of the Bloch wave-function can equivalently be written using the position-Hamiltonian commutator \cite{thonhauser2011theory},
\begin{equation}
    \ket{\partial_{j}u_m} = -\frac{i}{\hbar}\sum_{l \neq m}\ket{u_l}\frac{\expval{{u_l|[x_j,H]|u_m}}}{E_l-E_m}.\label{thon}
\end{equation} 
Following from this, the orbital magnetic moment operator takes the form
\begin{equation}\label{eq:mm}
    \hat{L}_i(\bm{k}) = \frac{\!ie\,\,}{2\hbar^3}\!\sum_{\substack{j,k,\\ p\neq n,m}}\epsilon_{ijk}\frac{
    [x_j,H]|u_p\rangle\langle u_p|[x_k,H]}{\bar{E}_{nm}-E_p}.
\end{equation}
$L_i$ is taken to be an operator acting on the states $\ket{u_n}$ and $\ket{u_m}$; these are assumed to be degenerate if the spin-orbit interaction is zero, but in general may not be exactly degenerate with finite spin orbit interaction\cite{RevModPhys.82.1959}.
Here $\bar{E}_{nm} = \frac{E_n+E_m}{2}$.\\ 
\\
If bands $n$ and $m$ are degenerate, so that $\frac{E_n + E_m}{2} = E_n = E_m$, then indeed 
the asymmetric mass term is the magnetization operator; since $\frac{i\hbar}{m}\pi_j = [x_j,H]$,
\begin{equation}\label{eq:mag_mm}
    \hat{L}_i = \frac{\!\!-ie\,\,}{2\hbar m^2}\sum_{\substack{j,k,\\ p\neq n,m}}\epsilon_{ijk}\frac{\pi_j|u_p\rangle\langle u_p|\pi_k}{\bar{E}_{nm}-E_p} =\frac{-ie}{\hbar}\epsilon_{ijk}\left(\frac{1}{m^*}\right)^{AS}_{jk}.
\end{equation}
The orbital angular momentum matrix for states $\ket{\xi}$ and $\ket{\bar{\xi}}$ is
\begin{eqnarray}\label{eq:m_mat}
&L_i =& \\ & \sum\limits_{jkl}\frac{\epsilon_{ijk}}{\bar{E}_{\xi\bar{\xi}}-E_l} 
\begin{pmatrix}
    \expval{\xi|\pi_j|u_l} \expval{u_l|\pi_k|\xi} & \expval{\xi|\pi_j|u_l} \expval{u_l|\pi_k|\bar{\xi}} \\
    \expval{\bar{\xi}|\pi_j|u_l} \expval{u_l|\pi_k|\xi} &  \expval{\bar{\xi}|\pi_j|u_l} \expval{u_l|\pi_k|\bar{\xi}}
\end{pmatrix}.&\nonumber
\end{eqnarray}
Then, similarly to Eq.\ref{eq:gs},
\begin{equation}\label{eq:gm}
    g_{Lij} = \frac{1}{2} \Tr(L_i\tilde\sigma_j).
\end{equation}

\section{Topological and entanglement properties of the \texorpdfstring{$\bm{g}$}{TEXT} -tensor} \label{sec:properties}
\subsection{Linear algebra of the \textit{g}-tensor} \label{sssection:asym}

The eigenvalues of $\bm{G}$ (Eqs.\ref{eq:DE},\ref{eq:Gg}) are the square of the singular values of $\bm{g}$.
The singular value decomposition of $\bm{g}$ gives
\begin{equation}\label{eq:SVD}
    \bm{g} = U\:\Sigma\:V.
\end{equation}
Given that $\bm{g}(\bm{k})$ is real, $U$ and $V$ are two $\bm{k}$-dependent, orthogonal, real matrices, unique except for columnwise (for $U$) and rowwise (for $V$) sign changes. $\Sigma$, also $\bm{k}$-dependent, is a unique positive semi-definite diagonal matrix of the singular values of $\bm{g}$ (taken to be in descending order along the diagonal). $\tiny{ \Sigma}_{xx} \geq \tiny{\Sigma}_{yy} \geq \tiny{\Sigma}_{zz}$ are the three scalar quantities, commonly know as the $\bm{g}$-factors.\\
\\
The columns of $U$ give the principal axes for the external magnetic field $\bm{B}$, such that the Zeeman splittings of the energy levels are
\begin{equation}
    \Delta E = \mu_{\bm{B}}\sqrt{\tiny{\Sigma}^2_{xx}B^2_{x} + \tiny{\Sigma}^2_{yy}B^2_y + \tiny{\Sigma}^2_{zz}B^2_z}.
\end{equation}
The total magnetic moment vector of the lower energy eigenstate in magnetic field $\bm{B}$ is
\begin{equation}
M^{ground}_i=\mu_B^2\Sigma^2_{ii}B_i/2\Delta E.
\label{eq:mtot}
\end{equation}

\subsection{Zeros of the \texorpdfstring{$\bm{g}$}{TEXT}-factor, and relation to entanglement}\label{sssec:continuity}
The $g$-tensor has been derived using the eigenvectors of the Hamiltonian, $H(\bm{k})$. We note that $H(\bm{k})$ can be chosen to vary continuously with $\bm{k}$. Consequently, the eigenvectors can also be chosen to be continuous in $\bm{k}$.
Thus, from the aforementioned procedure of calculating the $g$-tensor, it must be a continuous quantity with respect to $\bm{k}$. $det(\bm{g}_S)$ must also be continuous.\\
\begin{table}[h!]
    \begin{center}
        \begin{tabular}{ |c|c|c| } 
             \hline
             $T$ & $\textcolor{blue}{\Gamma_4} \otimes D_{1/2}$ = $\textcolor{red}{\Gamma_5} \oplus \Gamma_6 \oplus \Gamma_7$  & $-$\\ 
             \hline
             $T_h$ & $\textcolor{blue}{\Gamma_4^+} \otimes D_{1/2}$ = $\textcolor{red}{\Gamma_5^+} \oplus \Gamma_6^+ \oplus \Gamma_7^+$ & $\forall \overrightarrow{OA}$\\
             & $\textcolor{blue}{\Gamma_4^-} \otimes D_{1/2}$ = $\textcolor{red}{\Gamma_5^-} \oplus \Gamma_6^- \oplus \Gamma_7^-$ & \\
             \hline
             $T_d$ & $\textcolor{blue}{\Gamma_4} \otimes D_{1/2} = \textcolor{red}{\Gamma_6} \oplus \Gamma_8$ & $\overrightarrow{OA}|\bm{A} \in \{(0,0,k_z),$\\
             &$\textcolor{blue}{\Gamma_5} \otimes D_{1/2} = \textcolor{red}{\Gamma_7} \oplus \Gamma_8$ & $(k_x,0,0), (0,k_y,0),$\\
             & & $(k_x, k_y, k_z)\}$
             \\
             \hline
             $O$ & $\textcolor{blue}{\Gamma_4} \otimes D_{1/2} = \textcolor{red}{\Gamma_6} \oplus \Gamma_8$ &\\ &$\textcolor{blue}{\Gamma_5} \otimes D_{1/2} = \textcolor{red}{\Gamma_7} \oplus \Gamma_8$ & $-$\\
             \hline
             $O_h$ & $\textcolor{blue}{\Gamma_4^+} \otimes D_{1/2} = \textcolor{red}{\Gamma_6^+} \oplus \Gamma_8^+$ & $ \forall \overrightarrow{OA}$ \\ 
             & $\textcolor{blue}{\Gamma_5^+} \otimes D_{1/2} = \textcolor{red}{\Gamma_7^+} \oplus \Gamma_8^+ $ & \\
             & $\textcolor{blue}{\Gamma_4^-} \otimes D_{1/2} = \textcolor{red}{\Gamma_6^-} \oplus \Gamma_8^-$ & \\ 
             & $\textcolor{blue}{\Gamma_5^-} \otimes D_{1/2} =\textcolor{red}{\Gamma_7^-} \oplus \Gamma_8^-$ & \\
             \hline
        \end{tabular}
    \caption{
    Band labels for which $det(\bm{g}_S)$ is guaranteed to pass through 0 along the ray $\overrightarrow{OA}$. Groups of bands are labeled according the irreducible representations of the states in the set of bands at ${\bf k}=0$ as given in \cite{altmann1994point},\cite{PhysRev96280}.  The left-most column of this table indicates the points groups. In the central columns, the bands denoted in blue are bands without spin-orbit interactions. These bands split into the red and black bands when spin-orbit couplings are taken into account. The red bands are the bands with Kramers degeneracy at $\bm{\Gamma}$ for which $det(\bm{g}_S)$ changes its sign along any ray $\overrightarrow{OA}$. The right-most column of the table contains the directions in which a maximal spin-orbit entanglement is guaranteed for the degenerate red bands of the corresponding symmetry group.}
    \label{table:1}
    \end{center}
\end{table}
For the proof of the upcoming theorem.\ref{thm:thm1}, we will require energy-band pairs $\mathcal{B}(\bm{k}), \bar{\mathcal{B}}(\bm{k})$ for which accidental degeneracy is absent. We present the definition of the non-degeneracy of the bands, which immediately permits the statement of the theorem.\\
To introduce a notation for this, we recall (Eq.\ref{eq:AH}) the expression for $H_{SO}$ in Hartree atomic units:
\begin{equation}
    H_{SO} = \frac{\alpha^2}{4}\left(\nabla \bm{V} \times \bm{p}\right)\cdot \bm{\sigma}.\nonumber
\end{equation}
$\alpha$ is the fine structure constant, but for the sake of Definition \ref{def:ND} and Theorem \ref{thm:thm1}, we treat it as a running constant.\\
\\
Consider the bands of Table.\ref{table:1}. At $\bm{k} = 0$, the 
bands that we will focus on are labeled, for spin-orbit parameter $\alpha=0$, according to the irreducible representations in blue in the second column of Table.\ref{table:1}. For spin-orbit parameter $\alpha>0$, these bands spilt into bands with the (double group) irreducible representations in red and black. The bands with the red label are doubly degenerate whereas the ones in black may have a different degeneracy.\\
\\
We will require that there be no accidental degeneracies involving the band pair $\mathcal{B}(\bm{k}), \bar{\mathcal{B}}(\bm{k})$ dispersing from the states belonging to irrep. \textcolor{red}{$\Gamma_i$} (red irreps.~from Table.\ref{table:1}). This pair is degenerate at $\bm{k} = 0$, and sometimes degererate elsewhere (for some point groups and some directions in $\bm{k}$-space). Furthermore, this pair of bands is degenerate as $\alpha\rightarrow 0$ with the other bands $\mathcal{B}_{mult}(\bm{k})$ associated with the accompanying irrep $\Gamma_j$ in the table, i.e., in the pair \textcolor{red}{$\Gamma_i$}$ \oplus \Gamma_j$. Besides these, we will need to require that no other degeneracies occur, as captured by this definition ($\alpha_*>0$ is some fixed spin-orbit strength):
\\
\begin{definition}\label{def:ND} Non-degeneracy of Bands.\\
Band pair $\mathcal{B}(\bm{k}), \bar{\mathcal{B}}(\bm{k})$ is non-degenerate if $\forall \alpha$ such that $0 <\alpha < \alpha_*$, $\exists k_{ND}$ and $\Delta E > 0$ and all three conditions are true: 
1) $\forall \bm{k}, |\bm{k}| \,\leq\, k_{ND}$, $\forall \mathcal{B}'(\bm{k}) \notin \{\mathcal{B}(\bm{k}), \bar{\mathcal{B}}(\bm{k}), \mathcal{B}_{mult}(\bm{k})\}$, $\forall \mathcal{B}''(\bm{k}) \in \{\mathcal{B}(\bm{k}), \bar{\mathcal{B}}(\bm{k}), \mathcal{B}_{mult}(\bm{k})\}$, $|\mathcal{B}'(\bm{k}) - \mathcal{B}''(\bm{k})| > \Delta E$. 
2) $\forall \bm{k}, |\bm{k}| \,\leq\, k_{ND}$, $\mathcal{B}(\bm{k})$ and $\bar{\mathcal{B}}(\bm{k})$ are nowhere degenerate with $\mathcal{B}_{mult}(\bm{k})$.
3) $\forall \bm{k}, |\bm{k}|\!\!=\!\! k_{ND}$, $|\mathcal{B}(\bm{k}) - \mathcal{B}_{mult}(\bm{k})| > \Delta E$ and $|\bar{\mathcal{B}}(\bm{k}) - \mathcal{B}_{mult}(\bm{k})| > \Delta E$.
\label{def1}
\end{definition}
For any band-pair in a particular crystal the conditions of our definition are easily checked, and are seen clearly to be true for all the materials studied later in the paper. We can now proceed directly with the first main result:

\begin{theorem}\label{thm:thm1}
Given a crystal with any of the cubic point group symmetries: $O$, $O_h$, $T$, $T_h$ or $T_d$. Consider any pair of energy bands belonging to the irreducible representations of the symmetry group of the crystal, indicated in red from Table.\ref{table:1}, which have come from states with the blue irreducible representation for zero spin-orbit coupling. We consider band pairs that are nondegenerate in the sense of Definition \ref{def1}. Then for sufficiently small spin-orbit coupling, there exists a surface in $k$-space, enclosing the $\Gamma$ point, on which $det(\bm{g}_S(k)) = 0$, 
where $\bm{g}_S(k)$ is the spin contribution to the $g$-tensor.
\end{theorem}

\begin{proof}
Consider the bands occurring in the band structure of crystals with the point groups $O,\: O_h,\: T, \:T_h,\: T_d$, that belong, at ${\bf k}=0$, to the irreducible representations in red in Table.\ref{table:1}. The condition that the spin-orbit interaction is nonzero implies that these states have only a two-fold degeneracy, and not higher (accidental degeneracies are excluded by assumption). This two-fold degeneracy is precisely the Kramers degeneracy. 
We denote these Kramers-degenerate pairs as $\ket{\xi_{\Gamma}}, \ket{\bar{\xi}_{\Gamma}}$ where $\Gamma$ is any of the red labels of Table.\ref{table:1}.\\
\\
The spin $g$-tensor of these doubly degenerate bands are calculated using Eqs.\ref{eq:s_mat},\ref{eq:gs}. At ${\bf k}=0$ and infinitesimal spin-orbit coupling $\alpha=0^+$
the eigenvectors of these bands belonging to irrep $\textcolor{red}{\Gamma_i}$ are built up from so-called \cite{cardona2005fundamentals} ``$p$-like states" $\ket{X}, \ket{Y}, \ket{Z}$, mutually orthogonal, normalized orbital states that are often predominantly composed of atomic p-orbital states. For all our red irreducible representations $\Gamma$, symmetry guarantees (cf.~\cite{cardona2005fundamentals}, and spin harmonics tabulated in \cite{altmann1994point}) that these eigenstates have the form
\begin{align}
    \ket{\xi_{\Gamma}} &= \ket{\frac{1}{2}, \frac{1}{2}} = \frac{1}{\sqrt{3}} \left(\ket{Z}\ket{\uparrow} + \ket{X}\ket{\downarrow} + i\ket{Y}\ket{\downarrow}\right), \notag\\
    \ket{\bar{\xi}_{\Gamma}} &= \ket{\frac{1}{2}, -\frac{1}{2}} = \frac{1}{\sqrt{3}} \left(\ket{Z}\ket{\downarrow} - \ket{X}\ket{\uparrow} + i\ket{Y}\ket{\uparrow}\right).\label{eq:SPlike}
\end{align}
For finite $\alpha$ all states of irrep $\Gamma_i$ are mixed. Some are not of the form \eqref{eq:SPlike} because they arise from different non-relativistic states, e.g.. for group $T_d$, one has $\textcolor{blue}{\Gamma_1} \otimes D_{1/2} = \textcolor{red}{\Gamma_6}$. For $\alpha=0^+$ these additional states always have the simple form $\ket{S}\ket{\uparrow}$, $\ket{S}\ket{\downarrow}$. Because of the mixing of these states, the eigenstates will have the more general form
\begin{align}
    \ket{\xi_{\Gamma}} &= \ket{\frac{1}{2}, \frac{1}{2}}+\ket{\delta\xi_{\Gamma}},\notag\\
    \ket{\bar{\xi}_{\Gamma}} &= \ket{\frac{1}{2}, -\frac{1}{2}}+ \ket{\delta\bar{\xi}_{\Gamma}}. \label{eq:SPlikeplus}
\end{align}
But because of the non-degeneracy condition, these corrections $\ket{\delta\xi}$ can be calculated with perturbation theory\footnote{We take as given the convergence of this perturbation, which is essentially the same as the issue of the convergence of the $k\cdot p$ perturbation calculations as in Eqs.~(\ref{lutt}) and (\ref{thon}), which is well established.
}. All energy denominators $E_{n'}^{(0)}-E_n^{(0)}$ are guaranteed by Definition \ref{def:ND} to have a magnitude no less than a finite constant $\Delta E$. We conclude from this that norm$(\ket{\delta\xi})=O(\alpha/\Delta E)$. With this, an explicit calculation (Eqs.\ref{eq:s_mat},\ref{eq:gs}) for any pair of states of the form \eqref{eq:SPlikeplus} gives
\begin{equation}\label{eq:gs_explapp}
    \bm{g}_S =
    \begin{pmatrix*}[r]
        \frac{2}{3} & \,\,0 & 0 \\
        0 & \,\,\frac{2}{3} & 0 \\
        0 & \,\,0 & -\frac{2}{3}
    \end{pmatrix*}+O(\alpha/\Delta E),
\end{equation}
Thus $det(\bm{g_S(k=0)})=-\frac{8}{27}+O(\alpha/\Delta E),$
so that $det(\bm{g_S(k=0)})<0$ for sufficiently small $\alpha$.

It is interesting to note that $g_S=-2/3$ coincides with the application of the Land\'e expression with quantum numbers  $L=1, S=1/2, J=1/2$ \cite{ashcroft1976solid} \footnote{The Lande expression for $\bm{g_S}$ can be seen from eq.31.37 of \cite{ashcroft1976solid}, p. 654, which shows the two contributions to the $g$-tensor in the atomic limit. It is interesting to note that while it is common to say that these band states have quantum numbers $\bm{L} = 1, \bm{S} = \frac{1}{2}, \bm{J} = \frac{1}{2}$, the symmetries that this statement implies are only partially in force. In particular, while symmetry dictates that $g_S$ is indeed exactly the Lande value $-2/3$ for $\alpha\rightarrow 0$, }$g_L$ is not fixed to its Lande value $+4/3$.

We now consider the calculation of $g_S({\bf k})$ for $|{\bf k}|=k_{ND}$ (see Definition \ref{def:ND}). To evaluate Eq.\ref{eq:s_mat}, we need the Bloch eigenvectors of bands $\mathcal{B}(\bm{k}), \bar{\mathcal{B}}(\bm{k})$; we will denote this pair as $\ket{u_n(k)}$ and $\ket{u_m(k)}$. We again can apply perturbation theory:
\begin{align}\label{eq:pert_vec}
    \ket{u_{n/m}(k)} &= \ket{u_{n/m}^{(0)}(k)} - \frac{\alpha^2}{4} \notag\\
    &\!\!\!\!\!\!\!\!\!\!\!\!\!\!\!\!\!\!\sum_{n'\neq n, m} \frac{\expval{u_{n'}^{(0)}(k)|\left(\nabla \bm{V} \times \bm{p}\right)\cdot \bm{\sigma}|u_{n/m}^{(0)}(k)}}{E_{n'}^{(0)}-E_{n/m}^{(0)}} \ket{u_{n'}^{(0)}(k)} \notag\\
    &+ \mathcal{O}\left(\frac{\alpha^4}{16\Delta E^2}\right). 
\end{align}
And again, we use the fact that all energy denominators are by assumption greater than  $\Delta E$. Thus, the states $\ket{u_{n/m}(k)}$, for sufficiently small $\alpha$, are arbitrarily close to the eigenvectors with no spin orbit coupling $\ket{u_{n/m}^{(0)}(k)}$. For these the three eigenvalues of $g_S$ are $+2$, and thus $det(\bm{g}_S(k_{ND}))=+8$. Since $det(\bm{g}_S)$ is continuous in ${\bf k}$ and changes sign between the origin and the sphere of radius $k_{ND}$, there must be a closed surface within this sphere on which $det(\bm{g}_S(k_{ND}))=0$.
\end{proof}

Note that the $g$ factors describe an isotropic quantity at ${\bf k} = (0, 0, 0)$. Equation \ref{eq:gs_explapp} does not contradict this, as changing the phase convention of one state, $\ket{\bar{\xi}_{\Gamma}} \rightarrow - \ket{\bar{\xi}_{\Gamma}}$, gives, for $\alpha=0^+$,
\begin{equation}
    \bm{g}_S =
    \begin{pmatrix*}[r]
        -\frac{2}{3} & 0 & 0 \\
        0 & -\frac{2}{3} & 0 \\
        0 & 0 & -\frac{2}{3}
    \end{pmatrix*}
\end{equation}\\
\\
whose determinant is also $-\frac{8}{27}$, and whose magnetic moment $\bm{M}$ is also identical, as seen from Eq.\ref{eq:mtot}. A similar change of sign convention causes the eigenvalues of $\bm{g}_S$  at $k_{ND}$ to be $(-2, -2, 2)$, but with $det(\bm{g}_S) = +8$ unchanged.\\
A final word about Theorem.\ref{thm:thm1}: on the face of it, the theorem seems to contradict the seemingly obvious assertion that $g_S=2$ if there is no spin-orbit coupling, which would mean that $det(\bm{g_S})=+8$. But we see that the actual statement is more subtle: as $\alpha\rightarrow 0$, $det(\bm{g_S({\bf k})})=+8$ for almost all ${\bf k}$. ${\bf k}=0$ is special because of the emergent 6-fold symmetry of the multiplets when $\alpha=0$. Thus, the order of limits $\alpha\rightarrow 0$, ${\bf k}\rightarrow 0$ matters, and it is permitted that a surface $det(\bm{g_S({\bf k})})=0$ always exists, but this surface must collapse onto the origin as $\alpha \rightarrow 0$. 
\\
We noticed that, along the directions $\overrightarrow{OA}$ shown in Table.\ref{table:1}, states in these bands exhibit maximum entanglement at the point where $det(\bm{g}_S) = 0$. We will now prove this interesting fact. For this, we first show the general relation between the reduced density matrices of the electron spin, $\rho_S$ and $\bar{\rho}_S$ of the states $\ket{\xi}$ and $\ket{\bar{\xi}}$ respectively.

\begin{lemma}\label{lm:lm1}
    Consider the reduced density matrices, $\rho_S$ and $\bar{\rho}_S$ of the Kramers-paired bands corresponding to the irreducible representations in red in Table.\ref{table:1} along the directions $\overrightarrow{OA}$. Then,
    \begin{equation}\label{eq:rsrts}
        \bar{\rho}_S = \sigma_y\rho_S^T\sigma_y.
    \end{equation}
\end{lemma}
\begin{proof}
    Consider the eigenstates $\ket{\xi}$ and $\ket{\bar{\xi}}$ of Kramers-paired bands corresponding to the irreducible representations of inversion-symmetric groups and the group $T_d$. We discuss the proof in three parts as the underlying symmetry for each of these cases is different. The first part deals with the inversion-symmetric groups. For the group $T_d$, degenerate eigenstates occur only along the $\bm{\Delta}$- and the $\bm{\Lambda}$- directions. The second and the third parts deal with the two cases of $T_d$.\\
    \begin{enumerate}
            \item The eigenstates $\ket{\xi}$ and $\ket{\bar{\xi}}$ of bands in the inversion-symmetric groups are related by an inversion followed by the time-reversal operation $\theta$ \cite{PhysRev96280}, \cite{inui2012group}, having the form,
            \begin{equation}
                \ket{\xi} = \sum_{r} u(r) e^{i\bm{k}\cdot \bm{r}},\,\,\,\, \ket{\bar{\xi}} = \sum_{r} \theta(u(-r)) e^{i\bm{k}\cdot \bm{r}} .
            \end{equation}
            The reduced density matrices of these two states are, $\rho_S = \Tr_{orb}\rho = \Tr_{orb} \left(\ketbra{\xi}{\xi}\right)$ and $\bar{\rho}_S = \Tr_{orb} \bar{\rho} = \Tr_{orb} \left( \ketbra{\bar{\xi}}{\bar{\xi}}\right)$. Inversion is a space operation and therefore does not affect the spin part. Since $\rho_S$ is hermitian, $\bar{\rho}_S = \sigma_y \kappa \rho_S \kappa \sigma_y = \sigma_y\rho_S^T\sigma_y$ ($\kappa$ indicating complex conjugation)\cite{TR}.

            \item In the $\bm{\Lambda}$-direction of $T_d$, the states $\ket{\xi}$ and $\ket{\bar{\xi}}$ belong to the irreducible representation $\Lambda_4$ \cite{Koster1963ThePO} and are given by linear combinations of the basis states of the little group $C_{3v}$. The basis states of $C_{3v}$, for the first partner of the irreducible representation $\Lambda_4$, can be chosen to be the following combinations involving the s-like and p-like (cf. Eq.~\ref{eq:SPlike}) mutually orthogonal states,
            \begin{align}\label{eq:c3vbasis}
                \ket{B^{\Lambda_4}_1} &= \ket{S}\ket{\uparrow}, \notag \\
                \ket{B^{\Lambda_4}_2} &= \frac{1}{\sqrt{3}} \left(\ket{X}\ket{\uparrow} + \ket{Y}\ket{\uparrow} + \ket{Z}\ket{ \uparrow}\right), \notag \\
                \ket{B^{\Lambda_4}_3} &= \frac{1}{\sqrt{6}} \left( \ket{X}(\ket{\uparrow} + i\ket{\downarrow})  - \ket{Y}(\ket{\uparrow}  - \ket{\downarrow})\right), \notag \\
                &- \frac{(1+i)}{\sqrt{6}}\ket{Z}\ket{\downarrow}.
            \end{align}
            Using the projection operator formalism as defined in section 4.4 of \cite{dresselhaus2007group} where $\hat{P}^{\Gamma_n}_{kl}$, by definition, transforms one basis vector $\ket{l}$ into the basis vector $\ket{k}$ of the same representation $\Gamma_n$. That is, the eigenstate $\ket{\xi} = \alpha_1\ket{B^{\Lambda_4}_1} + \alpha_2\ket{B^{\Lambda_4}_2} + \alpha_3\ket{B^{\Lambda_4}_3}$ of the representation $\Lambda_4$ is transformed using the projector into $\ket{\bar{\xi}} = \hat{P}^{\Lambda_4}_{\bar{\xi}\xi}\ket{\xi}$. We see that this transformation involves a reversal of spin directions in these basis states. 
            
            Indeed, direct calculation, starting from
            \begin{equation}
                \bar{\rho} = P^{\Lambda_4}_{\bar{\xi}\xi}\rho P^{\Lambda_4\dagger}_{\bar{\xi}\xi},
            \end{equation}
            where $\bar{\rho}$ is the density matrix corresponding to the state $\bar{\xi}$,
             shows that this leads to the reduced density matrices $\rho_S = \Tr_{orb}(\rho)$ and $\bar{\rho}_S = \Tr_{orb}(\bar{\rho})$ being related as
            \begin{equation}
                \bar{\rho}_S = \left(
                \begin{array}{rr}
                    \rho_{S22} & -\rho_{S12}\\
                    -\rho_{S21} & \rho_{S11}
                \end{array}
                 \right)
            \end{equation}
            which is the component-wise form of Eq.\ref{eq:rsrts}. It is interesting to note that this derivation makes no use of the time reversal symmetry.

            \item The relevant representation $\Delta_5$ of the little group $C_{2v}$ of $T_d$ has four basis states formed out of the $s$- and $p$-like mutually orthogonal states:
    \begin{align}
        \ket{B^{\Delta_5}_1} &= \ket{S}\ket{\uparrow}, \notag \\
        \ket{B^{\Delta_5}_2} &= \ket{X}\ket{\downarrow}, \notag \\
        \ket{B^{\Delta_5}_3} &= \ket{Y}\ket{\downarrow}, \notag \\
        \ket{B^{\Delta_5}_4} &= \ket{Z}\ket{\uparrow}.
    `\end{align}
    Note the alternation of spin directions!
    The projector $\hat{P}^{\Delta_5}_{\bar{\xi}\xi}$, following the definition of \cite{dresselhaus2007group}, transforms the state $\ket{\xi} = \alpha_1\ket{B^{\Delta_5}_1} + \alpha_2\ket{B^{\Delta_5}_2} + \alpha_3\ket{B^{\Delta_5}_3} + \alpha_4\ket{B^{\Delta_5}_4}$ into $\ket{\bar{\xi}} = \hat{P}^{\Delta_5}_{\bar{\xi}\xi}\ket{\xi}$. Given the form of the eigenvector $\ket{\xi}$, its spin density matrix is diagonal,
    \begin{equation}
        \rho_S = 
        \begin{pmatrix}
            \rho_{S11} & 0\\
            0 & \rho_{S22}
        \end{pmatrix}\label{eq:start}
    \end{equation}
    Using the projector, we arrive at the state,
    \begin{equation}
        \bar{\rho}_{S} = \Tr_{orb} \left( \hat{P}^{\Delta_5}_{\bar{\xi} \xi} \rho \hat{P}^{\Delta_5\dagger}_{\bar{\xi} \xi}\right)
        = \begin{pmatrix}
            \rho_{S22} & 0 \\
            0 & \rho_{S11}
        \end{pmatrix},\label{eq:case3}
    \end{equation}
    \end{enumerate}
    consistent with Eq.\ref{eq:rsrts}.
\end{proof}
Note that Eq.\ref{eq:case3} implies the satisfaction of the condition Eq.\ref{eq:rsrts} for any choice of the partners of the representation:
for an arbitrary SU(2) rotation
\begin{equation}
    V =
    \begin{pmatrix}
        X & Y \\
        -Y^* & X^*
    \end{pmatrix},
\end{equation}

\begin{align}\label{eq:vrv}
    V\bar{\rho}_S V^\dagger &=
    \begin{pmatrix}
        c|X|^2-(1-c)|Y|^2 & (1 - 2c)XY \\
        (1 - 2c)X^*Y^* & -(c-1)|X|^2+c|Y|^2
    \end{pmatrix} \notag\\
    &=\sigma_y(V\rho_S V^\dagger)^T\sigma_y,
\end{align}
where $c=\rho_{S11}$ of Eq.\ref{eq:start}.\\
\\
Using the lemma, we can show that these two states are maximally spin-orbital entangled when $det(\bm{g}_S) = 0$.

\begin{theorem}\label{thm:thm2}
    Consider the reduced density matrices of the spin subsystem, $\rho_S, \bar{\rho}_S$, of the Kramers-paired bands corresponding to the irreducible representations in red along the directions $\overrightarrow{OA}$ as indicated in Table.\ref{table:1}. Then
    \begin{equation}
        det(\bm{g}_S) = 0 \implies \mathcal{S}(\rho_S) = \mathcal{S}(\bar{\rho}_S) = 1,
    \end{equation}
     where $\mathcal{S}(\rho_S) = -\Tr( \rho_S \log\:\rho_S)$ is the entanglement entropy of the reduced density matrix of the spin subsystem.
\end{theorem}

\begin{proof}
    The spin $g$-tensor of these states can be decomposed according to Eq.\ref{eq:SVD} as $\bm{g}_S = U\Sigma V$. We go to the coordinate system of the Bloch sphere such that in this transformed basis $V$ is rotated away, i.e., 
    $\bm{g}_S = U\Sigma$. This necessitates a change of basis of $\ket{\xi}$ and $\ket{\bar{\xi}}$ into $\ket{\xi'}$ and $\ket{\bar{\xi}'}$ by an $SU(2)$ transformation,
    \begin{equation}\label{eq:transform}
        \begin{pmatrix}
            \ket{\xi'} \\
            \ket{\bar{\xi}'}
        \end{pmatrix} = 
        \begin{pmatrix}
            X & Y\\
            -Y^* & X^*
        \end{pmatrix}
        \begin{pmatrix}
            \ket{\xi} \\
            \ket{\bar{\xi}}
        \end{pmatrix}
    \end{equation}
    where $X$ and $Y$ are complex amplitudes such that $|X|^2+|Y|^2 = 1$. Equation \ref{eq:vrv} shows that for any $X$ and $Y$ this rotation does not affect the relation between $\rho_S$ and $\bar{\rho}_S$. We choose $X, Y$ to reflect the mapping $SU(2) \longrightarrow SO(3)$ \cite{westra20082},
    \begin{equation}
        V^{-1} =
        \begin{pmatrix}
            \Re(X^2 - Y^2) & \Im(X^2 + Y^2) & -2\Re(XY) \\
            -\Im(X^2-Y^2) & \Re(X^2 + Y^2) & 2\Im(XY) \\
            2\Re(XY^*) & 2\Im(XY^*) & |X|^2 - |Y|^2
        \end{pmatrix}.
    \end{equation}\\
    The total density matrices of the time-reversal pair $\ket{\xi'}$ and $\ket{\bar{\xi}'}$ are $\rho = \ket{\xi'}\bra{\xi'},\:\bar{\rho} = \ket{\bar{\xi}'}\bra{\bar{\xi}'}$. Tracing out the orbital degrees of freedom leaves us with the reduced density matrices of the spin-subsystem, $\rho_S = \Tr_{orb} \rho = \ket{\xi'_S}\bra{\xi'_S},\: \bar{\rho}_S = \Tr_{orb} \bar{\rho} = \ket{\bar{\xi}'_S}\bra{\bar{\xi}'_S}$. \\
    The spin matrix from Eq.\ref{eq:s_mat} can be re-written as
    \begin{align}\label{eq:srot}
        S_i &= 
    \begin{pmatrix}
      \Tr \ket{\xi'}\bra{\xi'}\sigma_i & \Tr \ket{\xi'}\bra{\bar{\xi}'}\sigma_i \\
      \Tr \ket{\bar{\xi}'}\bra{\xi'}\sigma_i & \Tr \ket{\bar{\xi}'}\bra{\bar{\xi}'}\sigma_i,
    \end{pmatrix},
    \end{align}
    and the spin $g$-tensor has the form
    \begin{widetext}
    \begin{equation}\label{eq:grot}
    \bm{g}_S = 
        \begin{pmatrix}
            \Tr (\ket{\xi'}\bra{\bar{\xi}'}\sigma_x + \ket{\bar{\xi}'}\bra{\xi'}\sigma_x) \,\,\,\,& i\Tr (\ket{\xi'}\bra{\bar{\xi}'}\sigma_x - \ket{\bar{\xi}'}\bra{\xi'}\sigma_x) \,\,\,\,& \Tr_{spin} (\rho_S\sigma_x - \bar{\rho}_S\sigma_x) \\
            \Tr (\ket{\xi'}\bra{\bar{\xi}'}\sigma_y + \ket{\bar{\xi}'}\bra{\xi'}\sigma_y) \,\,\,\,& i\Tr (\ket{\xi'}\bra{\bar{\xi}'}\sigma_y - \ket{\bar{\xi}'}\bra{\xi'}\sigma_y) \,\,\,\,& \Tr_{spin} (\rho_S\sigma_y - \bar{\rho}_S\sigma_y) \\
            \Tr (\ket{\xi'}\bra{\bar{\xi}'}\sigma_z + \ket{\bar{\xi}'}\bra{\xi'}\sigma_z) \,\,\,\,& i\Tr (\ket{\xi'}\bra{\bar{\xi}'}\sigma_z - \ket{\bar{\xi}'}\bra{\xi'}\sigma_z) \,\,\,\,& \Tr_{spin} (\rho_S\sigma_z - \bar{\rho}_S\sigma_z) \\
        \end{pmatrix}.
    \end{equation}
    \end{widetext}
    The eigenvectors from Eq.\ref{eq:transform} along with the transformed spin matrix of Eq.\ref{eq:srot} give us the spin $g$-tensor in Eq.\ref{eq:grot}. Recall that $\bm{g}_S = U\Sigma$ which means that up to a sign factor, $det(\bm{g}_S) = det(\Sigma)$. Recall also that $\Sigma$ is the diagonal matrix of singular values in descending order. Since we have
    \begin{equation}
        det(\bm{g}_S) = det(\Sigma) = \Sigma_{xx}\Sigma_{yy}\Sigma_{zz} = 0,
    \end{equation}
    then at least $\Sigma_{zz} = 0$, so that all the entries of the third column of $\bm{g}_S)$ are zero. Thus, 
    \begin{align}\label{eq:rho}
        &\Tr_{spin} (\rho_S\sigma_x - \bar{\rho}_S\sigma_x) = 0,\notag \\
        &\Tr_{spin} (\rho_S\sigma_y - \bar{\rho}_S\sigma_y) = 0,\notag \\
        &\Tr_{spin} (\rho_S\sigma_z - \bar{\rho}_S\sigma_z) = 0, \notag \\
        &\implies \rho_S - \bar{\rho}_S = 0.
    \end{align}
  Now, combining $\rho_S=\bar{\rho}_S$ with $\bar{\rho}_S = \sigma_y \rho_S^T \sigma_y$ (lemma \ref{lm:lm1}) gives 
  \begin{equation}
       \rho_S = \bar{\rho}_S =\tfrac{1}{2}\mathbb{1},
      \end{equation}
    so that
    \begin{equation}
        \mathcal{S}(\rho_S) = \mathcal{S}(\bar{\rho}_S) = 1.  
    \end{equation}
\end{proof}

It should be noted that theorem \ref{thm:thm2} applies not only to the three cubic point groups $O_h, T_h$ and $T_d$, but also to various other non-cubic groups (which we look at no further here). As can be noticed in Table.\ref{table:1}, theorem \ref{thm:thm2} does not apply at all to the groups $T$ and $O$, because they do not have any directions in which pairs of states are degenerate, even taking time reversal into account. This has been verified using the method described in \cite{inui2012group}. 

We show that, due to theorem \ref{thm:thm1}, bands with certain symmetries necessarily have surfaces surrounding the $\Gamma$-point where $det(\bm{g}_S) = 0$. In the rest of the paper, we discuss the application of the theory of the $g$-tensor to the important semiconductors silicon (Si), germanium (Ge) and gallium arsenide (GaAs). We also show the complex topologies of the surfaces of both $\bm{g}_S$ and $\bm{g}=\bm{g}_L+\bm{g}_S$ for important valence and conduction bands with the relevant symmetry in each of the three cases mentioned in the lemma above.\\

\section{Tight-binding studies}

Computations of $g$-factors for important bands in silicon and germanium were carried out already long ago \cite{PhysRevLett.3.217,PhysRev.118.1534,PhysRevLett.6.683}. We take up these computations here, but with an emphasis on illustrating the topological features of the $g$-factor zeros in important bands, and the resulting spin-orbit entanglement properties of the band states. We use a simple tight binding approach, and thus do not strive for ab-initio accuracy. We use a model that has been used extensively for studying states in various nanostructures, but that has never been adjusted or tested for the magnetic properties that we study here. Nevertheless, our calculations give an impression of the general $g$-factor features that can be expected in these semiconductors.

\subsection{Model Hamiltonian}\label{ssec:hamiltonian}

We use here the $sp^3d^5s^*$ tight-binding model \cite{onen2019calculating,nemo_proj,nemo_url} which is seen to be quite accurate for energy eigenstates in Si and Ge. As originally discussed by Chadi \cite{PhysRevB.16.790}, the Hamiltonian contains diagonal on-site energies and off-diagonal interaction terms. The spin-orbit Hamiltonian contains $\expval{ p_{i,\sigma}|H_{SO}|p_{j,\sigma}}$ interactions between $p_x$ and $p_y$ orbitals of the same spin $\sigma$ and $\expval{ p_{i,\sigma}|H_{SO}|p_{k,\sigma'}}$ on-site interactions between $p_{x/y}$ and $p_z$ orbitals of opposite spins $\sigma, \sigma'$. The spin-orbit interaction between the d-orbitals is much smaller compared to that between the p-orbitals and is usually neglected. The parameters for Si and Ge have been taken from \cite{PhysRevB.79.245201}.
The known band structures of crystalline Si and Ge with spin-orbit coupling fits very precisely with this model (Fig. \ref{fig:BS}).
\onecolumngrid\
\begin{figure}[ht!]\
\begin{center}\
\includegraphics[width=0.5\textwidth]{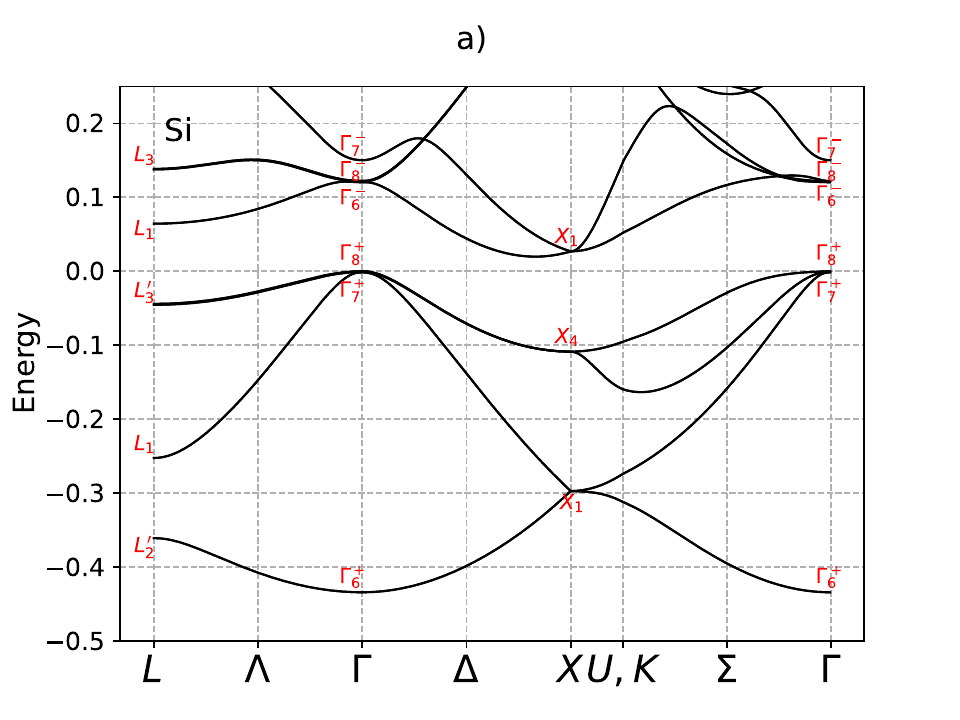}
\includegraphics[width=0.5\textwidth]{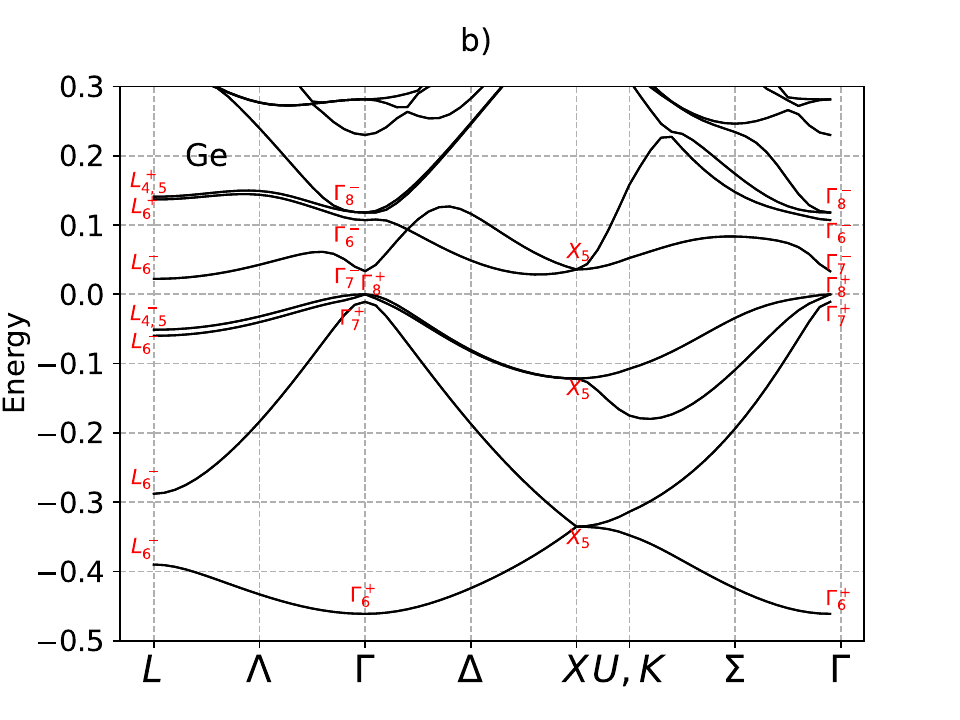}
\caption{(a) The band-structure of silicon. (b) The band-structure of germanium. The red labels are the irreducible representations of the band states at $\bm{\Gamma}$, taking spin-orbit coupling into account; see Table \ref{table:1}. Energies are in atomic Hartree units.}
\label{fig:BS}
\end{center}\
\end{figure}\
\twocolumngrid\

In the case of silicon, we see  that the bands $\Gamma^+_{5}$ and $\Gamma^-_{4}$ of Table.\ref{table:1} have tiny splittings because the spin-orbit interaction in silicon is relatively weak. At the $\Gamma$-point, these bands split into doubly degenerate Kramers pairs and four-fold degenerate bands, $\Gamma^+_{5} \rightarrow \Gamma_7^+ + \Gamma_8^+$ and $\Gamma^-_{4} \rightarrow \Gamma_6^- + \Gamma_8^-$. The splitting is more clearly visible in the case of germanium, the split-off band $\Gamma^+_7$ being relatively far from the band $\Gamma^+_8$. The gap between the bands $\Gamma^-_8$ and $\Gamma^-_6$ is also clearly seen.

\subsection{Treatment of momentum matrix elements}\label{sssec:Si_intra}

Following \cite{PhysRevB.63.201101} we approximate the momentum matrix elements as the gradient of the Hamiltonian ($\expval{\pi} \sim \expval{\nabla_{\bm{k}}H}$) plus certain non-zero intra-atomic contributions:
\begin{align}\label{eq:p_elements}
\expval{u_n(\bm{k})|\bm{\pi}|u_m(\bm{k})} &= \frac{m}{\hbar}\sum_{\alpha,\beta}a^{a*}_n a^{\beta}_m\nabla_{\bm{k}}\expval{\alpha|H|\beta} \notag\\
&+ \sum_{\alpha,\beta} a^{\alpha *}_n a^\beta_m\expval{\alpha|{\bm d}|\beta}(E_n-E_m).
\end{align}
The second term of Eq.\ref{eq:p_elements} contains the intra-atomic position operator \textit{\textbf{d}}. \\
\\ 
To calculate the orbital magnetic moment contribution $\bm{g}_L$ we must have values for the matrix elements $\bm{d}_{\alpha \beta} = \expval{\alpha|\bm{d}|\beta}$ with the orbital indices $\alpha, \beta$. $\bm{d}_{\alpha \beta}$ is generally not diagonal in the orbital index \cite{PhysRevB.63.201101}. The non-zero cases are determined by symmetry. The $\Delta l = \pm1$ selection rule tells us that the principal off diagonal contribution will be between $s$ and $p$ orbitals. These matrix elements are needed for describing the intra-atomic transitions. In previous tight-binding studies, these have been fit to the  optical properties of the bulk semiconductors \cite{PhysRevB.41.12002}.

Here, we choose instead to fit them to the Land{\'e} $g$-factor in the atomic limit. For isolated atoms the $g$-factor is, according to the Land{\'e} theory, 
\begin{align}
g_J &= \frac{1}{2}\left(1 + \frac{L(L+1) - S(S+1)}{J(J+1)}\right)   \\
&+ \frac{g_0}{2}\left(1 - \frac{L(L+1) - S(S+1)}{J(J+1)}\right)= +\frac{2}{3},  \notag 
\end{align}
where $g_0$ is the vacuum electron $g$-factor. Within our tight binding model in the atomic limit, this value of $g_J = +\sfrac{2}{3}$ occurs when $\expval{s|\bm{d}|p_i}_{Si} = 2.788$ for an isolated Si and $\expval{s|\bm{d}|p_i}_{Ge} = 2.535$ for an isolated Ge. Values are given in atomic units ($a_0$). For GaAs each species has its own intra-atomic contribution; fitting again to the Land{\'e} value gives  $\expval{s|\bm{d}|p_i}_{Ga} = 2.891, \expval{s|\bm{d}|p_i}_{As} = 2.455$. \\
\\
Having set up the tight-binding models, we proceed to calculate the two contributions to $\bm{g}_{tot}=\bm{g}_L+\bm{g}_S$. We focus on the evolution of the singular values of these quantities, and also on the scalar invariant which is the determinant. The following sections discuss the numerical manifestations of theorems \ref{thm:thm1} and \ref{thm:thm2}. We then discuss the complex topologies of the surfaces considered by these two theorems.

\subsection{Calculation of the \texorpdfstring{$\bm{g}$}{TEXT}-factors}

\subsubsection{Singular values of the \texorpdfstring{$\bm{g}$}{TEXT}-tensor}
The important eigenstate pairs for which the results of theorems \ref{thm:thm1} and \ref{thm:thm2} can be seen are $\ket{\xi_{\Gamma_7^+}}$-$\ket{\bar{\xi}_{\Gamma_7^+}}$ (the split-off valence band) and  $\ket{\xi_{\Gamma_6^-}}$-$\ket{\bar{\xi}_{\Gamma_6^-}}$ (a low lying conduction band). Such bands show the change in the sign of the $g$-factors quite clearly. We first look at the lowest conduction band in the case of silicon.
The plots in Fig. \ref{fig:singval} show the singular values of $\bm{g}_S$ and $\bm{g}_{tot}$ for silicon in a randomly chosen direction going outward from the $\Gamma$ point. Since we do not go along a symmetry direction, $\bm{g}$ is highly anisotropic, evolving to have three distinct singular values. For $\bm{g}_S$ (Fig. \ref{fig:singval}), at the $\Gamma$-point, $\Sigma_{xx}, \Sigma_{yy}, \Sigma_{zz}$ are $(\frac{2}{3}, \frac{2}{3}, \frac{2}{3})$, as expected from the above discussion. $det(\bm{g}_S)$  (red dashed curve) starts at $-\sfrac{8}{27}$, crosses zero at the point where at least one singular value goes to zero (exactly one, in this case), and asymptotes eventually to $+8$. 
\onecolumngrid\
\begin{figure}[h!]\
\begin{center}\
\includegraphics[width=0.5\textwidth]{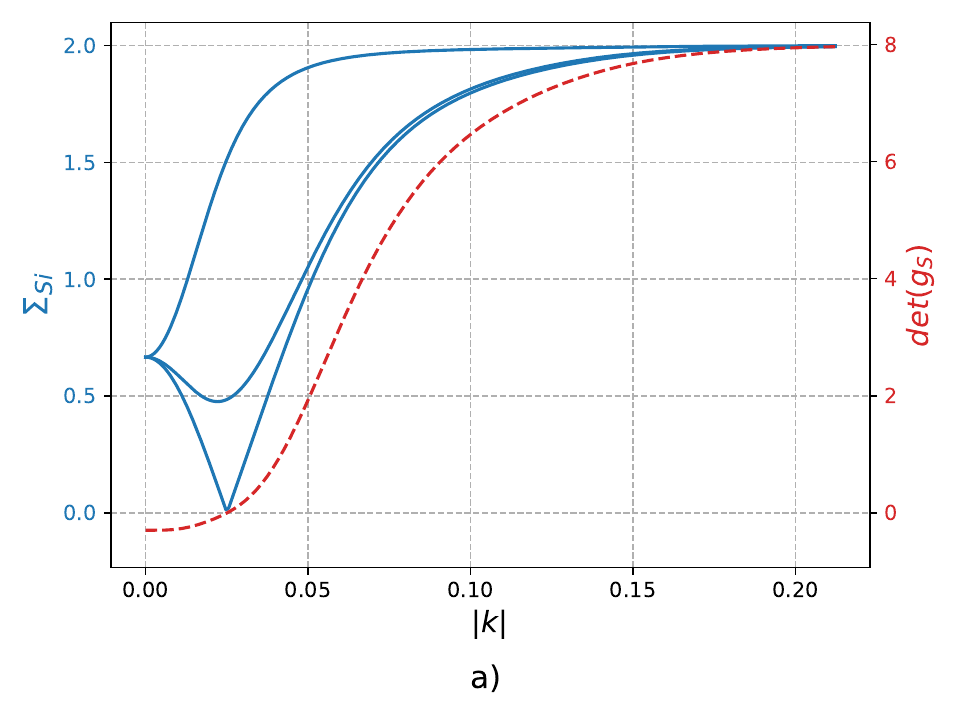}
\includegraphics[width=0.5\textwidth]{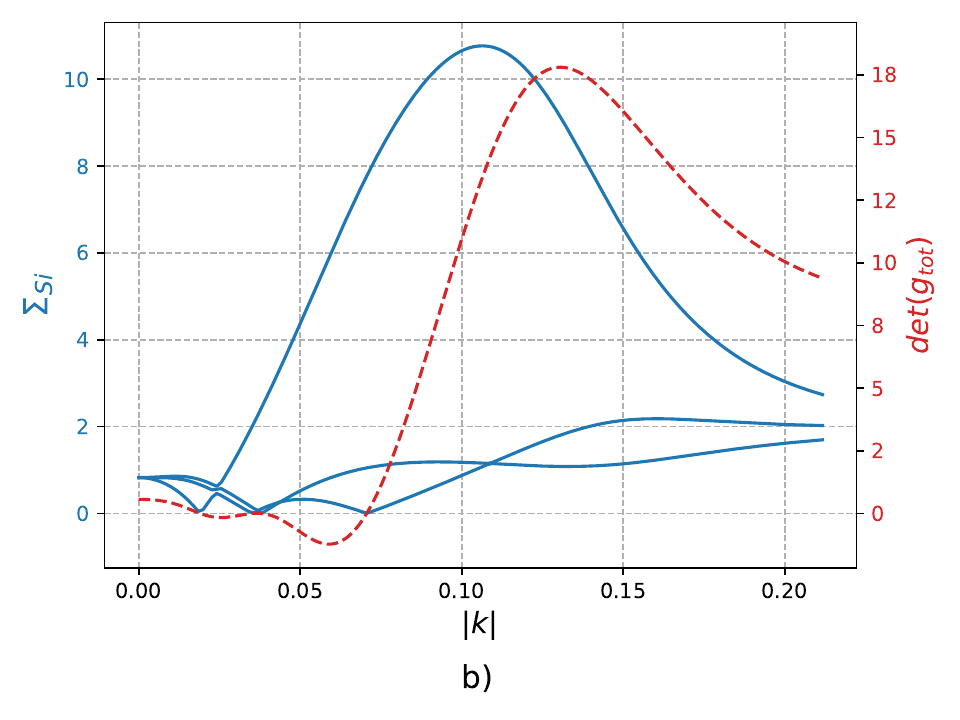}
\caption{Occurrence of zeros in $g$-factors, as seen in the singular values of the $g$-tensor (blue), for the first conduction band of Si. These singular values (by definition non-negative) are calculated in an arbitrary direction to highlight the anisotropy. 
The red dashed curve is the evolution of the determinant of the $g$-tensor. Panel a) shows the singular values and the determinant of only the spin contribution to the $g$-tensor, $g_S$. The singular values start at 2/3 and converge to 2 which is what we must observe for Si. Panel b) shows the singular values and the determinant of $\bm{g}_{tot}=\bm{g}_{L}+\bm{g}_{S}$. The change in the sign of the determinant is a clear indication of the change in the sign of the $g$-factors. $|k|$ is given in inverse atomic units, $a_0^{-1}$.}\
\label{fig:singval}
\end{center}\
\end{figure}
\twocolumngrid\

Although not guaranteed by our theorem, we see that $det(\bm{g}_{tot})$ also crosses zero (Fig. \ref{fig:singval}), in fact multiple times, corresponding to zeros in its singular values. It may also be remarked that the orbital magnetic moment contributions are quite large at places where adjacent bands approach the lowest conduction bands, causing one $g$-factor to be much larger than $2$.

\subsubsection{Determining the zero-crossing surface}\label{sssec:bisection}
According to our theorem, as we travel outward from ${\bf k}=0$ in any direction, at least one singular value of $\bm{g}_S$, $\Sigma_{zz}$, should go to zero where $det(\bm{g}_S)$ goes to 0, at a point we will call $\bm{k}_c$ (there may be an odd number of such points), assuming that the spin-orbit interaction is ``small enough". This is indeed what we observe for bands in both silicon and germanium. The zero-crossing occurs in all directions around the $\Gamma$-point creating one or multiple smooth surfaces surrounding it. These surfaces, given for band $n$ by $det(\bm{g}_S({\bf k}_c))=0$, can be as complex as Fermi surfaces, $E_n({\bf k}_F)=E_F$, and will have the same cubic symmetry. A technique similar to the ones used to find Fermi surfaces can be employed to determine them. We discuss briefly one such method that we implement to study their topology.\\
\\
Consider the singular value decomposition of $\bm{g}_S$ as defined in Eq.\ref{eq:SVD}. Since the $g$-tensor is a real square matrix, $U$ and $V$ can be chosen to be real orthogonal matrices. Then, since $det(\Sigma) \geq 0$,
\begin{equation}\label{eq:gen_kc}
    sgn(det(\bm{g}_S)) = det(U)\cdot det(V).
\end{equation}
The change in sign of the $g$-tensor can be determined by calculating the product of the determinants ($\pm 1$) of the real-valued orthogonal unitary matrices $U$ and $V$. As discussed in theorem.\ref{thm:thm2}, with a change of basis we set $V=\mathbb{1}$ so that $sgn(det(\bm{g}_S)) = det(U)$. But this can cause issues in the directions of high symmetry. The form in Eq.\ref{eq:gen_kc} is generic for all directions in $\bm{k}$-space. Binary search on $sgn(det(\bm{g}_S))$ can be a good choice of algorithm for finding its zero crossing. We find that it converges quite quickly and accurately to $\bm{k}_c$. Using the plots of the singular values, one can get a rough estimate for $\bm{k}_c$. Choosing that as a midpoint of a short interval, one may be able to come close to the value of $\bm{k}_c$ with an arbitrary precision. This is a method we implemented to enhance the accuracy of this detection.\\
\\
We also verify numerically that the eigenstates $\ket{\xi_{\Gamma_{7^+}}(\bm{k}_c)}$ and $\ket{\bar{\xi}_{\Gamma_{7^+}}(\bm{k}_c)}$ have maximal entanglement entropy. Note that all the other states of the form $\ket{\theta,\phi} = \cos(\theta)\ket{\xi_{\Gamma_{7^+}}} + e^{i\phi}\sin(\theta)\ket{\bar{\xi}_{\Gamma_{7^+}}}$ on this Bloch sphere have $S(\Tr_{orb}\ket{\theta,\phi}\bra{\theta,\phi}) < 1$. For example, the entanglement entropies of the cardinal states, $\ket{\pm} = \left[\frac{1}{\sqrt{2}}\left(\ket{\xi_{\Gamma_7^+}} \pm \ket{\bar{\xi}_{\Gamma_7^+}}\right)\right], \ket{\pm i} = \left[\frac{1}{\sqrt{2}}\left(\ket{\xi_{\Gamma_7^+}} \pm i\ket{\bar{\xi}_{\Gamma_7^+}}\right)\right]$ are $S(\Tr_{orb}\ket{\pm}\bra{\pm}) = 0.810$ and $S(\Tr_{orb}\ket{\pm i}\bra{\pm i}) = 0.824$, which are obviously far from the maximal value. \footnote{One could also compute the spin $g$-tensor using the density matrix formalism. We consider the spin matrix in the following form,
\begin{equation*}
    S_i = 
    \begin{pmatrix}
        \Tr(\rho\sigma_i) & \Tr(\ket{\bar{\xi}}\bra{\xi}\sigma_i) \\
        \Tr(\ket{\xi}\bra{\bar{\xi}}\sigma_i) & \Tr(\bar{\rho}\sigma_i)
    \end{pmatrix}.
\end{equation*}
The diagonal terms of $S_i$ are clearly a function of $\rho_S$ and $\bar{\rho}_S$. The off-diagonal terms that do not correspond to any density matrix can be developed using the quantum process tomography formula \cite{nielsen2002quantum}(Refer to eq(8.154) on page 391), where $\varepsilon$ is any linear operator:
    \begin{align*}
        \varepsilon \left(\ket{\xi}\bra{\bar{\xi}}\right) &= \varepsilon \left(\ket{+}\bra{+}\right) + i\varepsilon \left(\ket{+i}\bra{+i}\right) \notag \\
        &- \frac{1+i}{2}\varepsilon \left(\ket{\xi}\bra{\xi}\right) - \frac{1+i}{2}\varepsilon \left(\ket{\bar{\xi}}\bra{\bar{\xi}}\right)
    \end{align*}
    where $\ket{+} = \frac{1}{\sqrt{2}}\left(\ket{\xi}+\ket{\bar{\xi}}\right)$ and $\ket{+i} = \frac{1}{\sqrt{2}}\left(\ket{\xi}+i\ket{\bar{\xi}}\right)$ are cardinal states on the equator of the Bloch sphere.}

\section{Topology of det(\texorpdfstring{$\bm{g}$}{TEXT})=0 surfaces}
The surfaces in $k$ space determined by the equation $det(\bm{g}(\bm{k}))=0$, which may be calculated using the method described in  Sec. \ref{sssec:bisection}, can have complex topology. In this section we look at some of these surfaces for the first conduction and split-off (valence) bands for Si, the second conduction and split-off bands for Ge and the split-off bands of GaAs. These are the bands belonging to the interesting representations as discussed in lemma \ref{lm:lm1}. The calculations for GaAs have been done using an $sp^3$-band model \cite{Chadi1975TightbindingCO,tbgit}, which has been fitted to match valence-band properties. In cases where the adjacent bands approach the bands of our interest, we observe a multiplicity of these surfaces. Both Si and Ge show this multiplicity in various directions.

\subsection{Maximal entanglement surfaces (MES) of \texorpdfstring{$\bm{g}_S$}{TEXT}}
\begin{figure}[h!]
     \centering
         \includegraphics[width=0.55\textwidth]{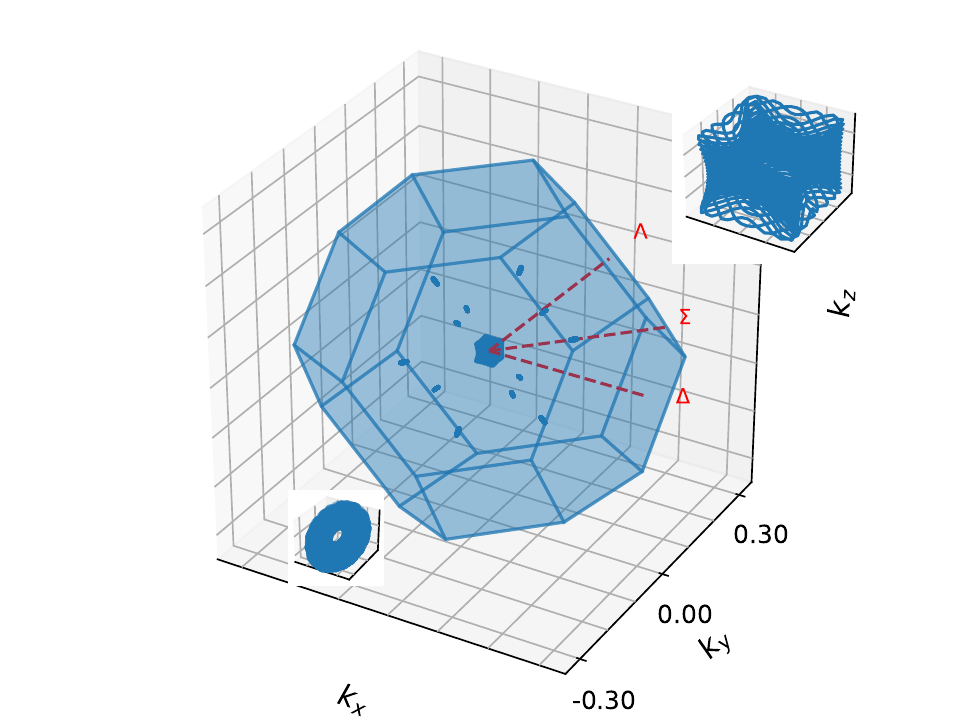}
     \caption{MES of silicon for the first conduction bands, with Brillouin zone boundary shown. The inset in the top right corner shows the innermost surface; grid spacing here is $10^{-2}\,a_o^{-1}$. This surface closely resembles a cuboid, and it is surrounded by flattened toruses in the $\bm{\Sigma}$ (110) directions. The inset in the bottom left corner is a zoom of this structure, with grid spacing here being $10^{-3}\,a_o^{-1}$.
     }
    \label{fig:si_89_gs}
\end{figure}
On the surfaces for which $det(\bm{g}_S) = 0$, the spin and orbital subspaces are maximally entangled for Si and Ge, as shown in theorem.\ref{thm:thm2}. In silicon, for the first conduction band, one sheet of the MES (Fig. \ref{fig:si_89_gs}) is a smooth cuboid with concave faces, but topologically a sphere. It extends out to about a tenth of the distance to the Brillouin zone boundaries (BZB), far away from the band minimum.\\
\begin{figure}[h!]
    \centering
        \includegraphics[width=0.5\textwidth]{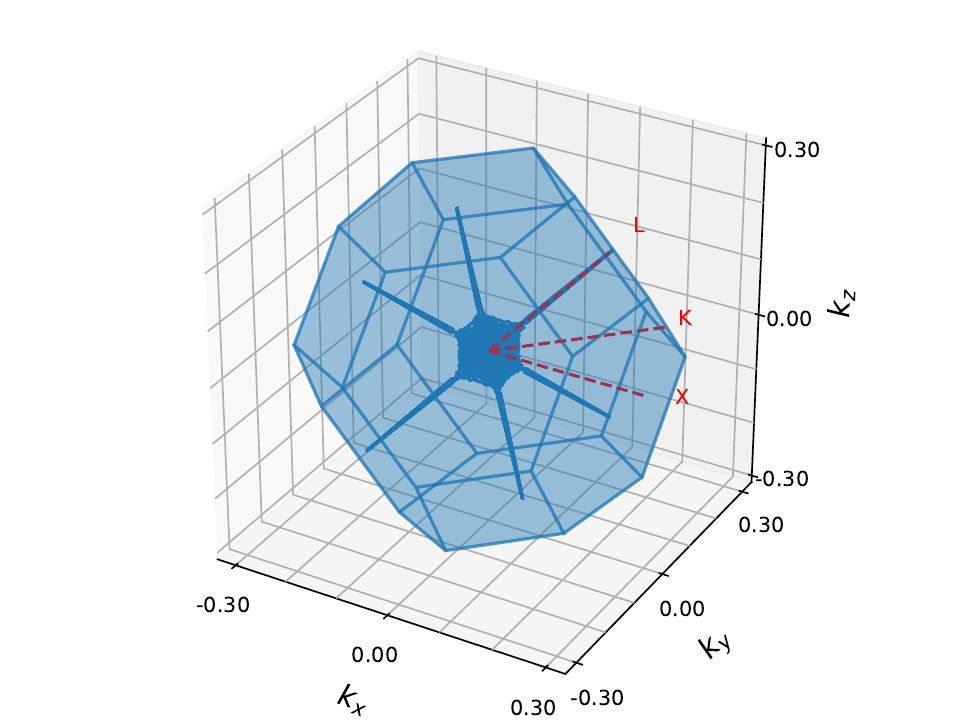}
    \caption{The germanium MES for the second conduction bands, with Brillouin zone boundary shown. The 
    surface resembles a cuboid with openings at the vertices. These openings connect to the central surfaces of the neighbouring zones via long tube like structures, whose thickness is exaggerated by a factor of 2 for visualization. 
    }
    \label{fig:ge1011_mes}
\end{figure}\\
A second family of sheets of this surface, topologically toruses, occurs in the $\bm{\Sigma}$-directions slightly over half-way from the BZB. They exist because the adjacent bands have anti-crossings with the first conduction band along the directions on which these toruses lie.\\
For Ge, the conduction bands exhibit an MES larger than that of Si, occupying almost $20\%$ of the BZ. This Ge MES is similar in shape to the central sheet of Si, except in the 111 directions. In these directions, the surface continues as thin rods parallel to the $\bm{\Lambda}$-directions, connecting the MES to the second Brillouin zone through the hexagonal faces (Fig. \ref{fig:ge1011_mes}). The topology of this MES is thus the same as that of the Fermi surface of copper \cite{ashcroft1976solid}. But compared with this Fermi surface, the connecting rods in this Ge MES are much longer and narrower. In the figure, the thickness of these rods has been exaggerated by a factor of 2 due to the thickness of the point plotting. The presence of these rods is again explained by the fact that adjacent bands lie at a comparable distance to these bands all along $\bm{\Lambda}$ (Fig. \ref{fig:BS}).\\
\\
The split-off band of silicon has an MES with smooth spikes extending into the $\bm{\Delta}$-direction, with the main part of the surface covering less than $10\%$ of the BZ (Fig. \ref{fig:gs_splitoff}(a)). The split-off bands of germanium form a MES almost resembling a sphere, also covering about $10\%$ of the distance of the BZB (Fig. \ref{fig:gs_splitoff}(b)).
\onecolumngrid\
\begin{figure}[h!]\
\begin{center}
\includegraphics[width=0.33\textwidth]{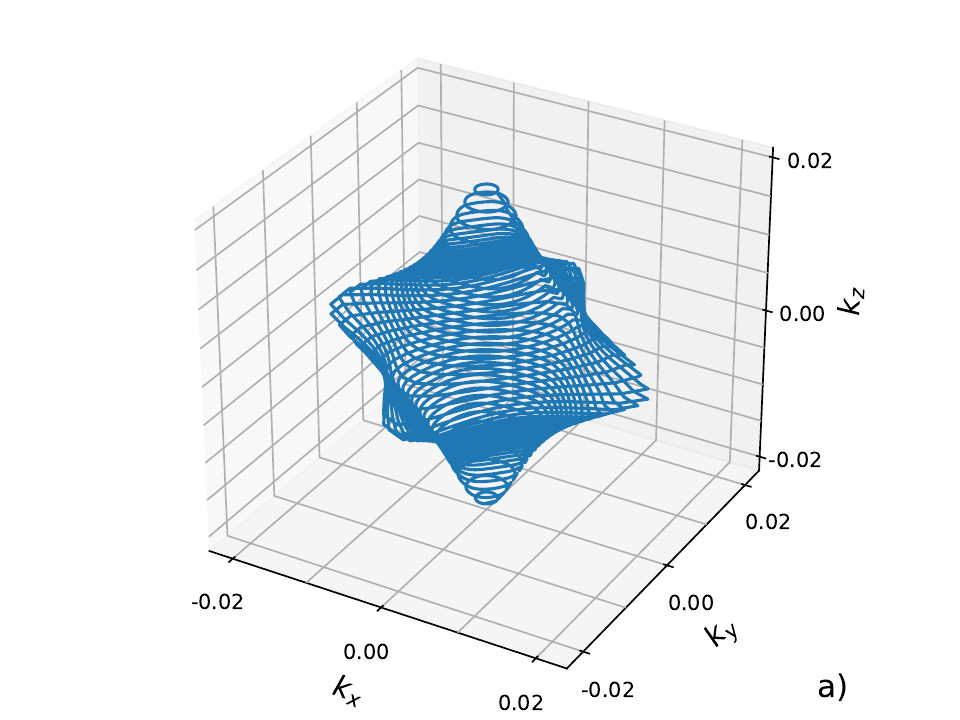}\hfill
\includegraphics[width=0.33\textwidth]{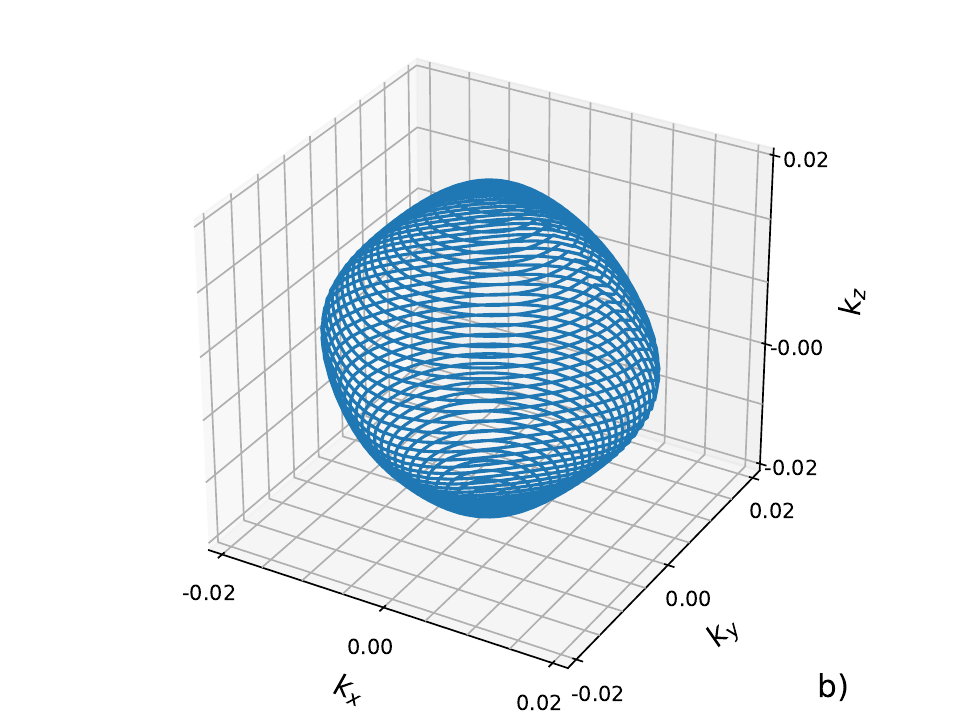}\hfill
\includegraphics[width=0.33\textwidth]{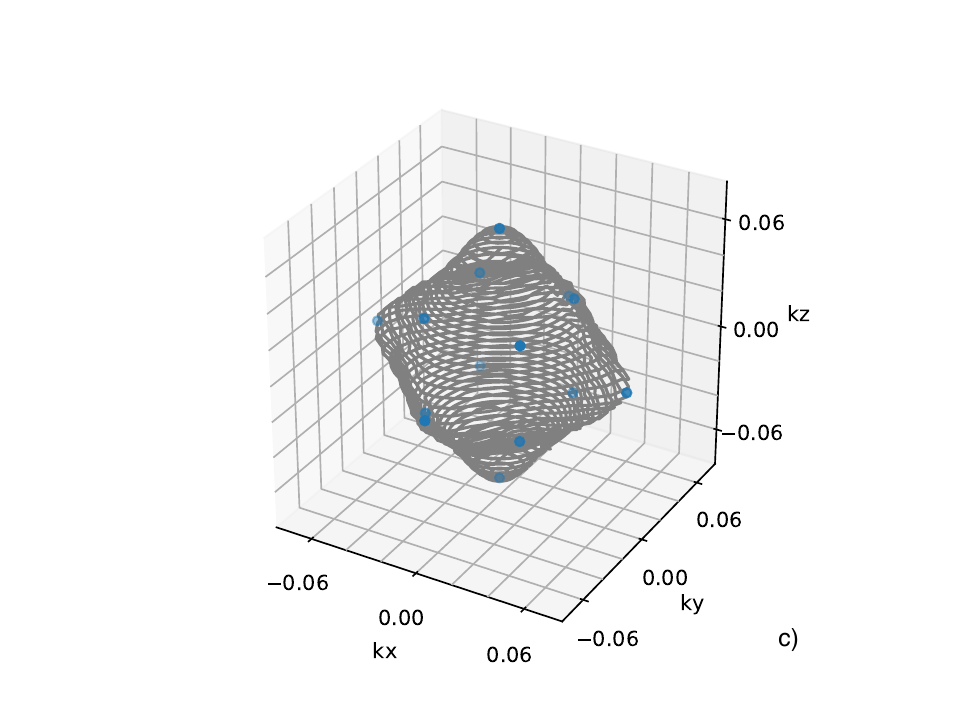}
\caption{MES surfaces of split-off bands (valence bands). (a) The surface of  Si. (b) The surface of  Ge. These surfaces are topologically spherical. (c) The split-off bands of GaAs exhibit maximum spin-orbit entanglement only along the $\bm{\Delta}$ and $\bm{\Lambda}$ directions. The grey surface is the surface where $det(\bm{g}_S) = 0$. The spin-orbit entanglement is maximal only at the blue points on this surface.}\
\label{fig:gs_splitoff}
\end{center}
\end{figure}
\twocolumngrid\
GaAs does not admit a full surface of maximum entanglement since, as shown in theorem \ref{thm:thm2}, we see maximal spin-orbit entanglement in only certain directions (tabulated in Table.\ref{table:1}). The surface where the determinant is 0 for the split-off valence band is plotted in Fig.  \ref{fig:gs_splitoff}(c), with the maximal-entanglement points indicated. 

\subsection{Zero-crossing surfaces of \texorpdfstring{$\bm{g}_{tot}$}{TEXT}}
Contributions of the orbital magnetic moment can have drastic effects on the singular values of $\bm{g}_{tot}$ and on the surfaces in $k$ space for which $det(\bm{g}_{tot}) = 0$. Singular values of $g_{tot}$ can be much larger than two when the energy difference between different bands is small. (Singular values of $g_S$ must be $\leq 2$.) We see large values of $g_{tot}$ already in Fig. \ref{fig:singval}.\\
\\
When approaching high symmetry directions, because of the occurrence of degeneracies in the singular values, our surfaces can pinch together, giving rise to touching conical structures as in the Fermi surfaces of semi-metals such as graphite \cite{Dresselhaus1964TheFS, PhysRevLett.108.140405}. We see such touchings for bands in both Si and Ge. These touchings are known as Lifshitz critical points \cite{voloviktopological}. This is associated with a quantum phase transition which occurs when the value of $det(\bm{g})$ is varied, in which these conical touchings open as gaps switching from the transversal to the longitudinal direction with respect to the axis on which the cones lie. 
This section provides further insights on the changes in topology of the surfaces of bands in both Si and Ge, admitting multiple Lifshitz critical points on the surface $det(\bm{g}_{tot}) = 0$, resembling the situation with some Fermi surfaces \cite{voloviktopological,PhysRevLett.128.076801}.\\
\\
In the case of the split-off valence band, the surfaces where $det(\bm{g}_{tot}) = 0$, for all three cases i.e Si, Ge and GaAs, are topologically spherical. As can be seen in Fig. \ref{fig:gtot_splitoff}, these surfaces do not exhibit any conical touchings. We see that these surfaces are generally larger than the ones for $det(\bm{g}_S) = 0$.
\onecolumngrid\
\begin{figure}[h!]\
\begin{center}\
\includegraphics[width=0.3\linewidth]{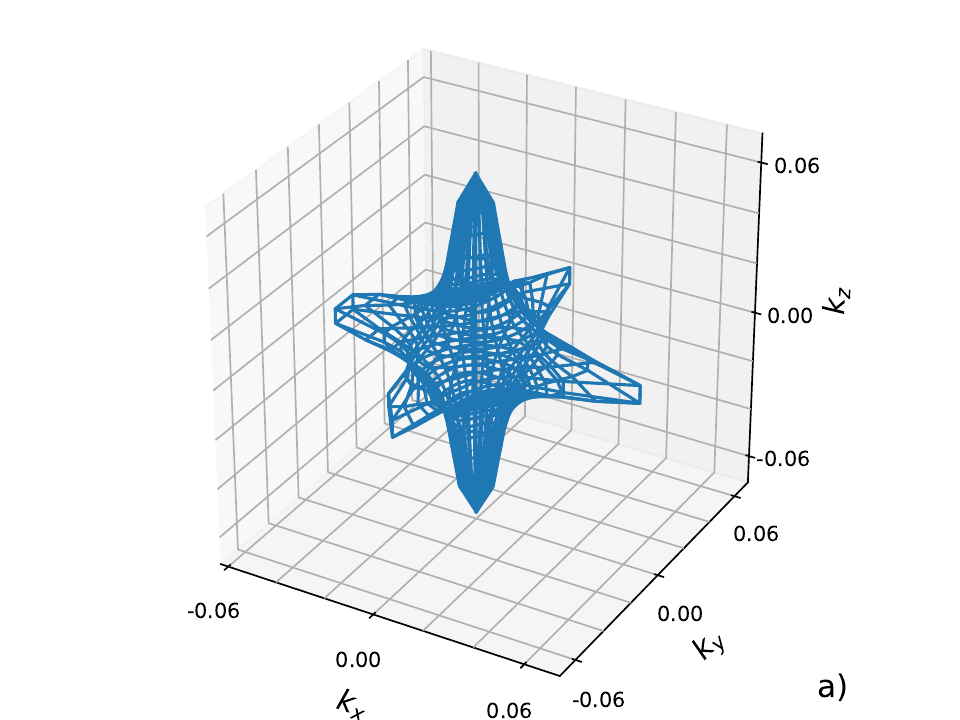}
\includegraphics[width=0.3\linewidth]{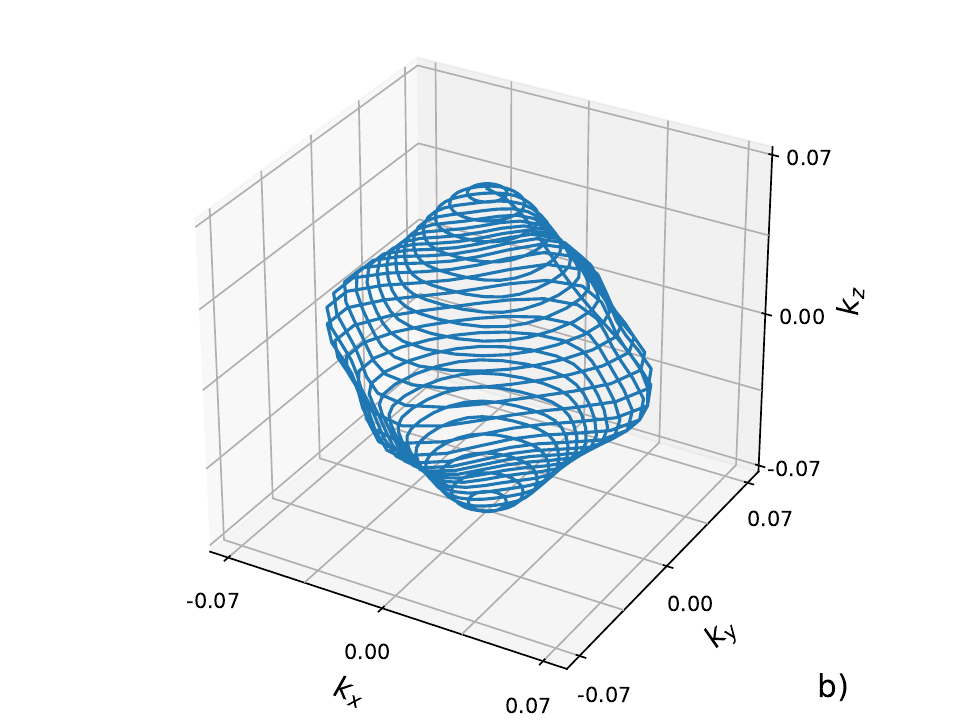}
\includegraphics[width=0.3\linewidth]{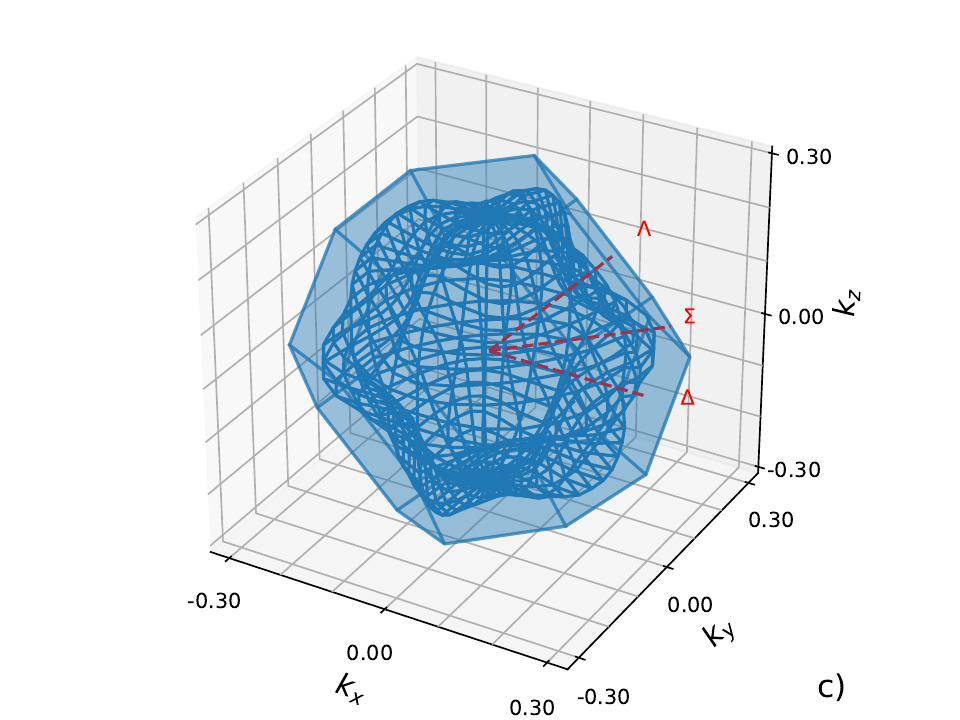}
\caption{Panels (a), (b), (c) show surfaces on which $det(\bm{g}_{tot}) = 0$ of the split-off bands for Si, Ge and GaAs respectively. Si and Ge surfaces are of comparable scales whereas the GaAs surface almost covers the entire BZ.}\
\label{fig:gtot_splitoff}
\end{center}\
\end{figure}\
\twocolumngrid\
All other cases exhibit nontrivial topologies.
For the first conduction band, Si exhibits a double surface around the $\bm{\Gamma}$-point. The inset of Fig. \ref{fig:gtot}(a) depicts a cross section of this complicated surface. As can been seen, it is approximately a square frame surrounded by a wing-like structure. In the full 3D picture, these have the form of two cuboids, one inside the other, extending outwards in the $\bm{\Sigma}$-direction and joining in the $\bm{\Delta}$ and $\bm{\Lambda}$-directions. The six joinings in the cross-section are six conical touchings, indicative of Lifshitz critical points.\\
\\
The occurrence of these Lifshitz critical points is specific to $det(\bm{g}_{tot}) = 0$ and arises from symmetry. The Lifshitz transition can be described as follows: For a slightly negative value of $det(\bm{g}_{tot})$, the conical touchings form gaps in the directions transverse to the symmetry axes. This causes the surface to have a nontrivial topology, because of the occurrence of holes surrounding the symmetry axes. The inner surface, in this case, is a surface of genus thirteen. In simple terms, it has the topology of a kind of wiffle ball. However, for a slightly positive value of $det(\bm{g}_{tot})$, the gaps instead open longitudinally along the symmetry axes. The surface transitions into a double sphere, i.e., whose genuses are both zero. The surrounding red and blue satellite regions have trivial topology (genus 0) that remains unaffected as $det(\bm{g}_{tot})$ is varied through zero.\\
\\
\onecolumngrid\
\begin{figure}[h!]\ 
\begin{center}\
\includegraphics[width=0.5\textwidth]{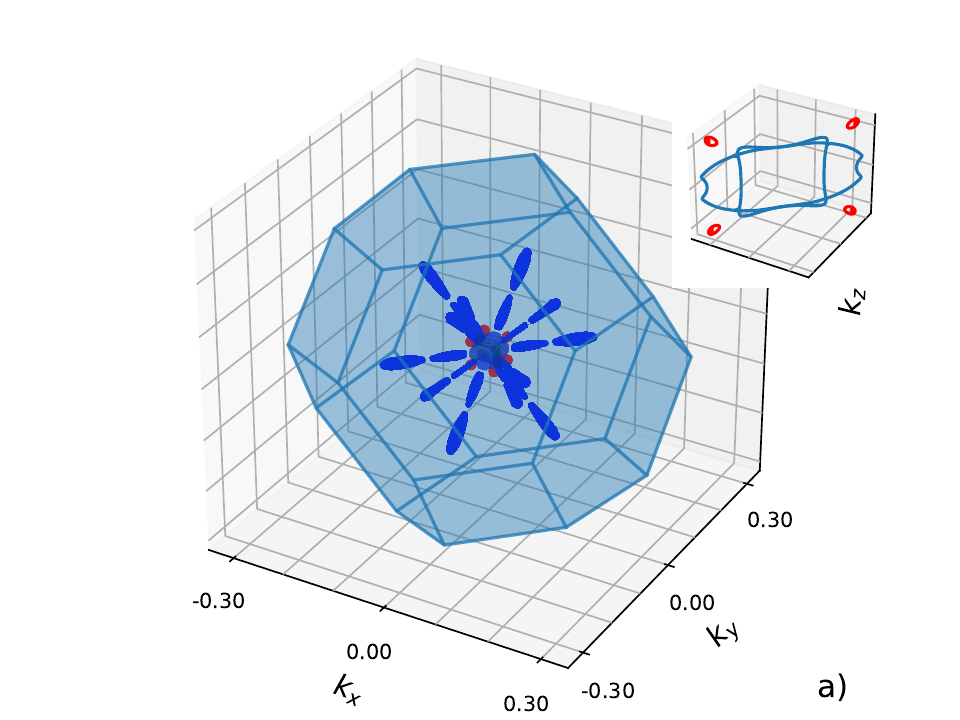}
\includegraphics[width=0.5\textwidth]{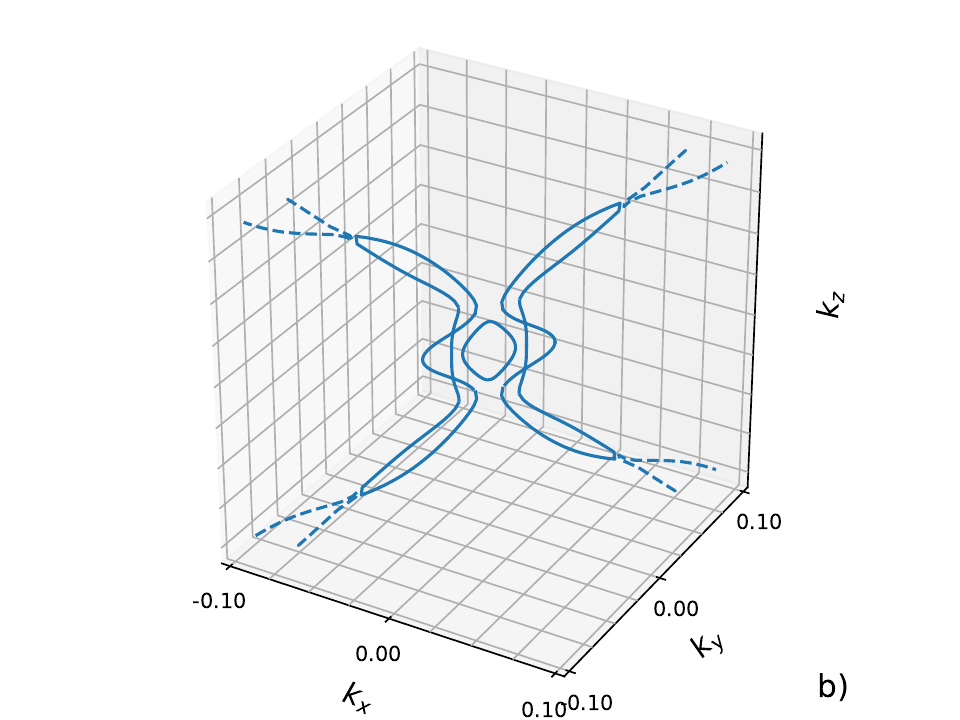}
\caption{(a) Visualization of the surfaces defined by $det(\bm{g}_{tot}) = 0$ for the first conduction band of Si. The central part is a double surface. The inset shows the cross-section of this central region and four spherical satellites in the $\bm{\Lambda}$-directions. One can clearly note the conical touchings in the $\bm{\Delta}$ and $\bm{\Lambda}$ symmetry directions. The Lifshitz transition at $det(\bm{g}_{tot})=0$ involves a change of genus of this surface from 0 to 13. Additional surfaces are present as elongated satellites in the $\bm{\Sigma}$ and $\bm{\Delta}$ directions. (b) Cross section of the $det(\bm{g}_{tot}) = 0$ Ge surface for the second conduction band. This cross section is informative because it shows four conical touching points (close to the central spherical surface) along the $\bm{\Lambda}$ directions. In this case as well, the Lifshitz transition changes the genus of the surface from 5 to 13.}\
\label{fig:gtot}
\end{center}\
\end{figure}\
\twocolumngrid\
The $det(\bm{g}_{tot})=0$ Ge surfaces of the second conduction band turn out to be topologically even more complex. The central surface is topologically a sphere. The second surface has a very rugged form. It has openings in the six $\bm{\Delta}$ directions. Figure \ref{fig:gtot}(b) shows a cross section of the entire set of surfaces in a plane containing the 111-direction. The dashed lines in the figure are additional surfaces with trivial topology, connecting into the second Brillouin zone. The points where these approach, but do not touch, the second surface are not special\\
One observes openings above and below the central surface, corresponding to holes in the second surface. One sees the cross section of the second surface in the figure as two separate regions on the left and the right, but these are fully connected in 3D. On both the left and the right one sees two intersection points, which are the cross sections of some of the conical touching points (of which there are eight overall), indicating that this surface exhibits Lifshitz criticality. $det(\bm{g}_{tot})$ also exhibits conical touchings on the second surface in the $\bm{\Delta}$ directions, but unlike for the Si case, these touchings occur only for a finite positive value of the determinant. This surface, for $1 >> det(\bm{g}_{tot}) > 0$ is one of genus 13. For a slightly negative value of $det(\bm{g}_{tot})$, the conical touchings all disconnect in the other way, causing this second surface to have genus 5.


\section{Conclusion and Outlook}
In this paper, we have provided a formalism for calculating the $g$-tensor and thus the $g$-factors, of any band at any given point in the Brillouin zone. We apply this formalism using tight-binding calculations for silicon, germanium, and gallium arsenide. The $g$-tensor has two contributions, one from the spin magnetic moment and one from the orbital magnetic moment. We provide a new derivation that links the old view, in which the orbital moment is described using the antisymmetric effective mass, to a new view which identifies this moment as originating from the Berry curvature of the energy band states, indicating orbital currents flowing inside a wavepacket made of these band states. We show that these points of view are fully consistent with each other.

We establish that, for any semiconductors with crystal point-group symmetries $O_h$ or $T_d$, bands of specific symmetry must exhibit Bloch eigenstates with maximum spin-orbital entanglement. Because of the weakness of the spin-orbit interaction in Si, the occurrence of this maximal entanglement occurs quite close to the $\bm{\Gamma}$-point for this material. This is quite far from the conduction band minimum along the $\bm{\Delta}$ direction, where this entanglement is quite small.

We hope that this work will establish a new mindset in thinking about the crystal g-factor. In no case that we have looked at is the g-factor ``boring", with just small variations from the non-relativistic value. New considerations, mixing together symmetry and topology in a novel way, determine the very interesting structure that we see in the crystal g-factors. This could be a harbinger of new effects to be looked for in various nanostructured devices. We mention only one idea, that the ``wiggle well" \cite{mcjunkin2022sige} technique can be used to selectively make states from particular parts of the crystal Brillouin zone relevant in nanostructures. This technique could be used to tease out information about, e.g., Lifshitz transitions manifested by the invariants of $\bm{g}_{tot}$. Finally, our work offers a new take on the sign of the g-factor: we see it playing the role of a new $\mathbb{Z}_2$ topological quantum number, indicating the occurrence of gap-closing quantum phase transitions of a novel character. In short, we see the g-factor as a promising subject for continuing investigations of both a practical and a fundamental nature.

\section{Acknowledgements}
We acknowledge fruitful discussions with our colleague Fabian Hassler on some aspects of group theory. We acknowledge support from the Deutsche Forschungs-Gemeinschaft (DFG, German Research Foundation) under Germany’s Excellence Strategy - Cluster of Excellence Matter and Light for Quantum Computing (ML4Q) EXC 2004/1 - 390534769.

\bibliography{bib_ref_ar}

\end{document}